\documentclass[conference]{IEEEtran}
\usepackage{amsmath}
\usepackage{amssymb}
\usepackage{algorithm}
\usepackage{amsthm}
\usepackage{graphicx}
\usepackage{float}
\newtheorem{theorem}{Theorem}
\newtheorem{lemma}{Lemma}
\newtheorem{assumption}{Assumption}

\newtheorem{proposition}{Proposition}
\usepackage{amsmath}
\usepackage{amssymb}
\usepackage{algorithm}
\usepackage{algpseudocode}
\usepackage{graphicx}  
\usepackage{subcaption} 
\usepackage{amsmath}   
\usepackage{amsmath,amssymb,amsfonts}
\usepackage{graphicx}
\usepackage{textcomp}
\usepackage{xcolor}
\usepackage{cite}
\usepackage{bm} 
\usepackage{algorithm}
\usepackage{algpseudocode}
\usepackage{amsmath}
\usepackage{amsfonts}
\usepackage{subcaption}

\usepackage{pgfplots}
\pgfplotsset{compat=1.17}

\title{A Risk-Aware Adaptive Robust MPC with Learned Uncertainty Quantification}
\author{\IEEEauthorblockN{Mingcong Li}
\IEEEauthorblockA{School of Automation\\
Beijing Institute of Technology\\
Beijing, China\\
limingcong0627@163.com}}

\begin{document}
\maketitle

\begin{abstract}
Solving chance-constrained optimal control problems for systems subject to non-stationary uncertainties is a significant challenge.    Conventional robust model predictive control (MPC) often yields excessive conservatism by relying on static worst-case assumptions, while standard stochastic MPC methods struggle when underlying uncertainty distributions are unknown a priori.  This article presents a Risk-Aware Adaptive Robust MPC (RAAR-MPC) framework, a hierarchical architecture that systematically orchestrates a novel synthesis of proactive, learning-based risk assessment and reactive risk regulation.  The framework employs a medium-frequency risk assessment engine, which leverages Gaussian process regression and active learning, to construct a tight, data-driven characterization of the prediction error set from operational data.  Concurrently, a low-timescale outer loop implements a self-correcting update law for an adaptive safety margin to precisely regulate the empirical risk and compensate for unmodeled dynamics.  This dual-timescale adaptation enables the system to rigorously satisfy chance constraints with a user-defined probability, while minimizing the conservatism inherent in traditional approaches.   We formally establish that the interplay between these adaptive components guarantees recursive feasibility and ensures the closed-loop system satisfies the chance constraints up to a user-defined risk level with high probability.   Numerical experiments on a benchmark DC-DC converter under non-stationary parametric uncertainties demonstrate that our framework precisely achieves the target risk level, resulting in a significantly lower average cost compared to state-of-the-art robust and stochastic MPC strategies.
\end{abstract}


\section{INTRODUCTION}

Model Predictive Control (MPC) has established itself as a cornerstone of modern control theory, demonstrating remarkable success in handling multivariable systems with complex dynamics and operational constraints across a wide array of industrial applications \cite{rawlings2017model, mayne2016robust, forbes2015model}. The efficacy of MPC stems from its core principle: repeatedly solving a finite-horizon optimal control problem online. This optimization-centric nature, however, renders its performance fundamentally dependent on the accuracy of the prediction model. The unavoidable presence of model uncertainties and external disturbances in real-world systems has thus spurred a rich and diverse field of research dedicated to robust MPC.

Traditional robust MPC addresses uncertainty through a worst-case lens, seeking to provide deterministic guarantees for constraint satisfaction and stability across all possible uncertainty realizations \cite{kothare1996robust, boyd1994linear}. This paradigm has led to powerful formulations based on dynamic programming \cite{scokaert1998min} and sophisticated convex optimization techniques, where complex min-max problems can be cast as structured programs like Quadratically Constrained Quadratic Programs (QCQP) or solved using semi-infinite programming \cite{diehl2007formulation, ganguly2024exact}. A key distinction within these methods is the treatment of feedback, leading to a well-established dichotomy between open-loop and closed-loop formulations \cite{lee1997worst, bemporad2002explicit}. Despite their theoretical rigor, worst-case approaches are often criticized for their inherent conservatism, as control actions are dictated by the most extreme, and often improbable, uncertainty scenarios. This conservatism can substantially degrade nominal performance and may even render the control problem infeasible.

To overcome this limitation, Stochastic Model Predictive Control (SMPC) offers a more nuanced and often more practical alternative. Instead of demanding absolute constraint satisfaction, SMPC recasts the problem using chance constraints, which require constraints to be satisfied with a user-specified, high probability \cite{lorenzen2017constraint, mesbah2016stochastic}. This probabilistic framing is particularly well-suited for applications where occasional, minor constraint violations are acceptable, such as in building climate control or chemical process management \cite{paulson2018nonlinear, bradford2020stochastic}. SMPC aims to systematically balance performance optimization against the risk of constraint violation, thereby enabling operation closer to the true operational limits and improving overall system efficiency. The central challenge within SMPC then becomes the tractable reformulation and solution of the chance-constrained optimal control problem.

The literature on SMPC is largely divided into two main schools of thought: analytical reformulations and sampling-based methods. Analytical approaches aim to convert probabilistic chance constraints into deterministic ones. For linear systems with Gaussian noise, this can often be done exactly \cite{farina2016stochastic, lorenzen2017constraint}. A prominent and effective paradigm in this category is the tube-based MPC \cite{mayne2011tube, blanchini1999set}. Here, the system state is confined within a "tube" around a nominal path, and the constraints are then tightened based on the tube's dimensions. While initially developed with fixed tubes, more advanced methods construct tubes based on incremental Lyapunov functions or control contraction metrics to reduce conservatism \cite{kohler2021computationally, manchester2017control}. A particularly relevant development involves the use of Probabilistic Reachable Sets (PRS), which can be computed offline via sampling to serve as probabilistic tubes, providing a bridge between the two main SMPC approaches \cite{hewing2020scenario}. However, the effectiveness of these methods depends critically on the characterization of the uncertainty, and they often assume stationary noise statistics, which may not hold in practice.

On the other hand, sampling-based methods, most notably the \textbf{scenario approach}, approximate the chance-constrained problem by enforcing constraints only for a finite number of randomly drawn uncertainty samples, or "scenarios" \cite{calafiore2006scenario, tempo2005randomized}. Rooted in statistical learning theory \cite{vidyasagar1997theory}, this technique transforms the stochastic problem into a deterministic convex program, for which strong probabilistic guarantees on the feasibility of the resulting solution can be established \cite{calafiore2013robust, bernardini2009scenario, campi2011sampling}. The scenario approach has been successfully extended to handle non-convex problems \cite{campi2018general} and sophisticated scenario management techniques, such as conditional scenario generation \cite{gonzalez2023conditional}, have been developed to improve its efficiency. Nevertheless, a key limitation of the standard scenario approach is its high computational cost for achieving low violation probabilities and its open-loop nature within the prediction, which can struggle to guarantee recursive feasibility without modifications \cite{schildbach2014scenario}.

Recently, the integration of machine learning has opened new avenues for data-driven MPC, particularly for adapting to uncertainty online. Gaussian Processes (GPs), with their inherent ability to provide uncertainty estimates alongside predictions, have proven to be a valuable tool \cite{rasmussen2006gaussian, li2024data}. A significant advancement in this direction is the work by Capone et al. \cite{capone2025online}, which uses a GP regression framework to directly learn the relationship between constraint-tightening parameters and the resulting satisfaction probability, enabling online adaptation. Beyond GPs, other learning techniques like neural networks and quasi-interpolation have been used to synthesize explicit MPC feedback laws offline, offering a-priori guarantees on the approximation error and enabling microsecond-level online evaluation \cite{ganguly2025explicit, parisini1995receding}.

Despite these considerable advances across different paradigms, a critical challenge remains: the development of control strategies that can robustly and efficiently adapt to non-stationary uncertainties. Real-world systems are rarely subject to static noise; instead, disturbance characteristics and system parameters often change over time. Existing methods, including many learning-based ones, typically assume stationarity and may therefore react slowly or inadequately to such changes, leading to transient periods of poor performance or constraint violations. There is a clear need for a framework that can intelligently assess risk in real-time and adapt its level of robustness in a systematic, provably safe manner.

This paper introduces a Risk-Aware Adaptive Robust MPC (RAAR-MPC) framework, a novel hierarchical architecture specifically designed to address this challenge. Our method systematically orchestrates a synthesis of proactive, learning-based risk assessment and reactive, experience-driven risk regulation. At its core, RAAR-MPC employs a dual-timescale adaptation mechanism. A medium-frequency loop utilizes a GP-based active learning engine to proactively identify critical uncertainty scenarios and construct a tight, data-driven Learned Prediction-Error Set (LPES). Concurrently, a low-frequency, outer loop implements a self-correcting update law for an adaptive safety margin, which precisely regulates the empirical risk based on closed-loop performance and compensates for unmodeled or non-stationary dynamics. This unique interplay enables the system to rigorously satisfy chance constraints with a user-defined probability, while minimizing the conservatism inherent in traditional approaches, even in the face of significant, time-varying uncertainties. We formally establish that this dual-adaptive architecture guarantees recursive feasibility and ensures the closed-loop system satisfies the specified chance constraints with high probability.

The main contributions of this work are threefold:
\begin{enumerate}
    \item \textit{A Novel Dual-Adaptive Robust MPC Architecture:} We propose a new framework that systematically integrates proactive, learning-based risk assessment with reactive, experience-driven adaptation. This architecture decouples the computationally intensive task of uncertainty quantification from the real-time control loop, enabling intelligent adaptation without compromising the speed of the MPC solve time.
    \item \textit{An Intelligent Risk-Informed Uncertainty Characterization:} We introduce a GP-based risk engine that leverages an Upper Confidence Bound (UCB) criterion to efficiently discover critical uncertainty scenarios. This leads to the construction of a Learned Prediction-Error Set (LPES) that provides a tight, non-parametric, and data-driven characterization of the propagated uncertainty, thereby reducing the conservatism inherent in traditional worst-case methods.
    \item \textit{A Rigorous Theoretical Framework with Formal Guarantees:} We provide formal proofs for the key properties of the proposed RAAR-MPC scheme. We establish recursive feasibility by construction and demonstrate that the controller achieves probabilistic constraint satisfaction with a quantifiable guarantee inherited from the LPES. Furthermore, we demonstrate the closed-loop stability of the entire system, including the physical state and the adaptive margin, by analyzing an augmented Lyapunov function, thus ensuring that the system remains bounded in expectation.
\end{enumerate}

The remainder of this paper is organized as follows. Section~\ref{sec:problem_formulation} presents the system modeling and formal problem formulation. Section~\ref{sec:raar_mpc_framework} details the proposed RAAR-MPC methodology, elaborating on the two core modules: the Online Risk Assessment Engine and the Dual-Layer Adaptive Robust Control Law. In Section~\ref{sec:theoretical_analysis}, we provide a formal analysis of the proposed framework, establishing its key theoretical properties. Section~\ref{sec:numerical_example} presents a numerical example on a benchmark DC-DC converter to demonstrate the efficacy and performance of our approach in comparison with state-of-the-art methods. Finally, Section~\ref{sec:conclusion} concludes the paper and outlines directions for future research.

\section{Problem Formulation}
\label{sec:problem_formulation}

We consider a discrete-time linear system subject to both time-varying parametric uncertainties and additive disturbances, described by the following state-space model:
\begin{equation}
    x_{k+1} = A_k x_k + B_k u_k + G d_k,
    \label{eq:true_system}
\end{equation}
where $x_k \in \mathbb{R}^{n_x}$ is the system state, $u_k \in \mathbb{R}^{n_u}$ is the control input, and $d_k \in \mathbb{R}^{n_d}$ represents an unknown, time-varying additive disturbance. The system matrices $(A_k, B_k)$ are themselves uncertain, potentially non-stationary, and are assumed to belong to a known compact set $\Delta$, i.e., $(A_k, B_k) \in \Delta$. The true realization of these matrices and the disturbance sequence $\{d_k\}$ are not known \textit{a priori}. We assume that the state and input must satisfy the following polytopic constraints for all $k \ge 0$:
\begin{equation}
    x_k \in \mathcal{X} = \{x \mid C_x x \le c_x\}, \quad u_k \in \mathcal{U} = \{u \mid C_u u \le c_u\},
    \label{eq:constraints}
\end{equation}
where $\mathcal{X}$ and $\mathcal{U}$ are compact sets that contain the origin in their respective interiors.

The control objective is to ensure that these constraints are satisfied probabilistically in the long run, despite the significant, time-varying nature of the uncertainty. This requirement is formalized as a chance constraint on the empirical frequency of constraint violations. For a given constraint function $h(x_k, u_k) \le 0$ representing one of the polytopic constraints in \eqref{eq:constraints}, we require:
\begin{equation}
    \lim_{T \to \infty} \frac{1}{T} \sum_{k=0}^{T-1} \mathbb{P}[h(x_k, u_k) > 0] \le \delta,
    \label{eq:chance_constraint}
\end{equation}
where $\mathbb{P}[\cdot]$ denotes the probability over the random realizations of uncertainties, and $\delta \in (0, 1)$ is a user-defined risk tolerance that specifies the maximum acceptable violation rate. The objective is to design a control law that minimizes a given performance cost while rigorously satisfying this chance constraint. This paper proposes a novel framework to address this challenge by adaptively learning the characteristics of the uncertainty online.

\section{The Risk-Aware Adaptive Robust MPC (RAAR-MPC) Framework}
\label{sec:raar_mpc_framework}


\subsection{Tube-Based Robust MPC with Adaptive Tightening}
\label{sec:tube_mpc}

The core of our framework is a tube-based robust Model Predictive Control (MPC) strategy. This approach decomposes the true system state and control input into a nominal component, which is optimized by the MPC, and an error component, which accounts for the effects of all uncertainties. At any time step $k$, the true state $x_k$ and input $u_k$ are defined as:
\begin{equation}
    x_k = z_k + e_k, \quad u_k = v_k + K_e e_k,
    \label{eq:state_decomposition}
\end{equation}
where $(z_k, v_k)$ are the nominal state and input, $e_k$ is the state error, and $K_e$ is a pre-computed stabilizing feedback gain for the error dynamics.

This decomposition allows us to separate the system dynamics into two parts. The nominal dynamics, used for prediction and optimization within the MPC, are given by:
\begin{equation}
    z_{k+1} = A_{\text{nom}} z_k + B_{\text{nom}} v_k,
    \label{eq:nominal_dynamics}
\end{equation}
where $(A_{\text{nom}}, B_{\text{nom}})$ is a known nominal model, which may be a mean or a simplified representation of the true system matrices. By substituting the decomposition \eqref{eq:state_decomposition} into the true system dynamics \eqref{eq:true_system} and subtracting the nominal dynamics \eqref{eq:nominal_dynamics}, we obtain the closed-loop error dynamics:
\begin{align}
    e_{k+1} &= (A_{\text{nom}} + B_{\text{nom}}K_e) e_k + (A_k - A_{\text{nom}})z_k \nonumber \\
    &\quad + (B_k - B_{\text{nom}})v_k + (B_k - B_{\text{nom}})K_e e_k + G d_k.
    \label{eq:error_dynamics}
\end{align}
This equation reveals that the evolution of the error $e_k$ is driven by a complex combination of the additive disturbance $d_k$ and the parametric model mismatch, $(A_k - A_{\text{nom}}, B_k - B_{\text{nom}})$, coupled with the nominal trajectory $(z_k, v_k)$.

To ensure robust constraint satisfaction, the error $e_k$ must be confined to a robust positive invariant (RPI) set for all time. A traditional approach would use a fixed, worst-case RPI set, leading to significant conservatism. Our key innovation is to define a Total Uncertainty Set, $\mathcal{U}_{\text{total}}(t)$, that is adapted online. This set is a composite structure, formed by the Minkowski sum of two distinct components:
\begin{equation}
    \mathcal{U}_{\text{total}}(t) = \mathcal{S}(t) \oplus \mathcal{B}(t),
    \label{eq:total_uncertainty_set}
\end{equation}
where:
\begin{itemize}
    \item $\mathcal{S}(t)$ is the Learned Prediction-Error Set (LPES), an outer approximation of the structured, predictable component of the error. It is constructed from data by a medium-frequency learning loop, as will be detailed in Section~\ref{sec:medium_freq_loop}. We represent this set by its time-varying, axis-aligned bounding box, characterized by support vectors $s_k(t) \in \mathbb{R}^{n_x}_{\ge 0}$.
    \item $\mathcal{B}(t)$ is the Adaptive Safety Margin, which accounts for unstructured or unmodeled errors not captured by the LPES. It is represented by a ball of radius $\beta_t$, where $\beta_t \in \mathbb{R}_{\ge 0}$ is a scalar safety margin updated by a low-frequency risk regulation loop, detailed in Section~\ref{sec:low_freq_loop}.
\end{itemize}

\begin{figure}[!t]
    \centering
    \includegraphics[width=0.85\columnwidth]{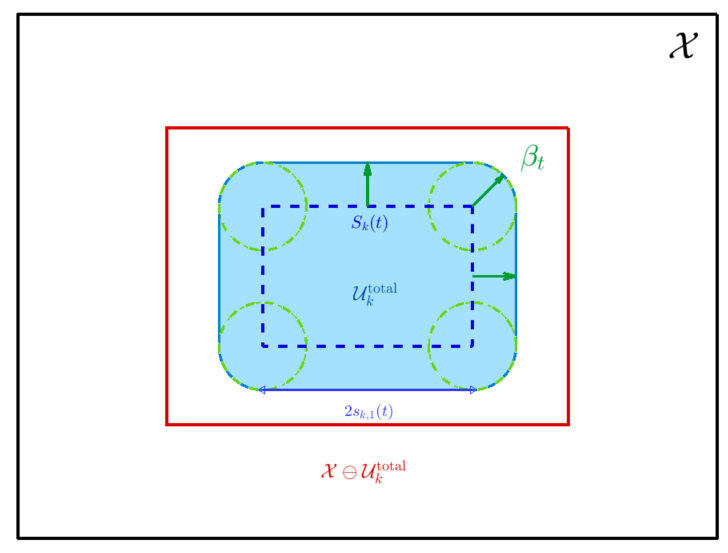}
    \caption{Visualization of the robust constraint tightening. The nominal state $z_k$ must remain within the tightened set $\mathcal{X} \ominus \mathcal{U}_k^\text{total}$, where $\ominus$ denotes the Pontryagin set difference. The total uncertainty set, $\mathcal{U}_k^\text{total}(t) = S_k(t) \oplus B(t)$, is the Minkowski sum of the axis-aligned Learned Prediction-Error Set $S_k(t)$ (a polytope, here a rectangle) and the Adaptive Safety Margin $B(t)$ (an $\ell_\infty$-ball of radius $\beta_t$). This construction robustly guarantees that the true state $x_k = z_k + e_k$ will satisfy the original state constraint $x_k \in \mathcal{X}$ for any error $e_k \in \mathcal{U}_k^\text{total}(t)$.}
    \label{fig:tightening}
\end{figure}

With these definitions, we can formulate the inner-loop robust MPC optimization problem. At each time step $t$, given the current state $x_t$, the LPES support vectors $\{s_{k|t}\}_{k=0}^{N-1}$, and the adaptive margin $\beta_t$, we solve the following convex Quadratic Program (QP) to find the optimal nominal control sequence $\mathbf{v}_t^* = \{v_{k|t}^*\}_{k=0}^{N-1}$:
\begin{subequations}
\label{eq:robust_mpc_qp}
\begin{align}
    \min_{\mathbf{z}_t, \mathbf{v}_t, \varepsilon_t} \quad & \sum_{k=0}^{N-1} \left( \|z_{k|t}\|_{Q}^2 + \|v_{k|t}\|_{R}^2 \right) + \|z_{N|t}\|_{P}^2 + \rho \varepsilon_t \\
    \text{s.t.} \quad & z_{0|t} = x_t, \\
    & z_{k+1|t} = A_{\text{nom}}z_{k|t} + B_{\text{nom}}v_{k|t}, \quad \forall k \in \{0, \dots, N-1\}, \\
    & z_{k|t} \in \mathcal{X} \ominus \mathcal{U}_{\text{total},k}(t), \quad \forall k \in \{0, \dots, N-1\}, \label{eq:state_tightening} \\
    & v_{k|t} \in \mathcal{U} \ominus K_e\mathcal{U}_{\text{total},k}(t), \quad \forall k \in \{0, \dots, N-1\}, \label{eq:input_tightening} \\
    & z_{N|t} \in \mathcal{X}_f, \\
    & \varepsilon_t \ge 0.
\end{align}
\end{subequations}
Here, $\mathbf{z}_t$ and $\mathbf{v}_t$ are the sequences of nominal states and inputs over the prediction horizon $N$. $Q$, $R$, and $P$ are positive semi-definite weighting matrices, $\mathcal{X}_f$ is a terminal set, and $\varepsilon_t$ is a slack variable with a large penalty $\rho \gg 0$ to ensure recursive feasibility. The symbol $\ominus$ denotes the Pontryagin set difference.

The core of the robust formulation lies in the constraint tightening in \eqref{eq:state_tightening} and \eqref{eq:input_tightening}. For the polytopic constraints defined in \eqref{eq:constraints}, the tightened state constraints are explicitly formulated for each row $i$ of the matrix $C_x$ as:
\begin{equation}
    C_{x,i} z_{k|t} \le c_{x,i} - \sup_{e \in \mathcal{U}_{\text{total},k}(t)} (C_{x,i} e) - \varepsilon_t.
    \label{eq:state_tightening_explicit}
\end{equation}
Given the axis-aligned and ball-shaped structure of our uncertainty sets, this supremum can be computed efficiently. Specifically, the tightening term becomes $\|C_{x,i}\|_1 (s^{\max}_{k|t} + \beta_t)$, where $s^{\max}_{k|t}$ is the maximum component of the LPES support vector $s_{k|t}$. A similar formulation applies to the input constraints.

The robust MPC controller \eqref{eq:robust_mpc_qp} guarantees constraint satisfaction for any error realization within the assumed total uncertainty set $\mathcal{U}_{\text{total}}(t)$. However, the performance and feasibility of this controller critically depend on the choice of the tightening parameters $s_k(t)$ and $\beta_t$. Overly conservative (large) values will shrink the feasible set, leading to poor performance or even infeasibility, while overly optimistic (small) values will result in frequent constraint violations. Therefore, the central challenge is to develop a systematic methodology for the online co-design of the LPES $s_k(t)$ and the safety margin $\beta_t$. The objective of this co-design is to satisfy the long-term chance constraint \eqref{eq:chance_constraint} with high precision while minimizing conservatism. The multi-timescale learning architecture designed to address this challenge is detailed in the subsequent sections.

\subsection{Multi-Timescale Adaptive Mechanism}
\label{sec:adaptive_mechanism}

The robust MPC controller described in the previous section relies on two adaptive quantities: the Learned Prediction-Error Set (LPES), $\mathcal{S}(t)$, and the adaptive safety margin, $\beta_t$. We now detail the hierarchical learning mechanism responsible for their online synthesis. This mechanism comprises two distinct loops operating on different timescales.

\subsubsection{Medium-Frequency Loop: Intelligent Risk Assessment and LPES Construction}
\label{sec:medium_freq_loop}

The Learned Prediction-Error Set (LPES), $\mathcal{S}(t)$, provides a tight, data-driven characterization of the structured component of the prediction error. It is updated periodically, for instance, every $M$ control steps, through a sophisticated risk assessment process. This process moves beyond static, worst-case assumptions by proactively identifying and quantifying the most critical uncertainty scenarios from operational data. The procedure consists of defining a criticality metric, learning a surrogate model for this metric, actively discovering high-criticality scenarios, and finally constructing the LPES from high-fidelity simulations. These components are detailed sequentially below.

\paragraph{A Lyapunov-Based Criticality Metric.}
The foundation of our risk assessment is a mathematically rigorous metric that quantifies the criticality of any potential uncertainty realization. Criticality is defined not by the magnitude of an uncertainty, but by its potential to destabilize the closed-loop error dynamics, thereby posing the greatest risk to constraint satisfaction.

Let an uncertainty realization over a prediction horizon of $N$ steps be denoted by the tuple $\zeta = (\mathbf{d}, \mathbf{\Delta})$, where $\mathbf{d} = \{d_0, \dots, d_{N-1}\}$ is the disturbance sequence and $\mathbf{\Delta} = \{(A_0, B_0), \dots, (A_{N-1}, B_{N-1})\}$ is the sequence of time-varying system matrices. For a given stabilizing feedback gain $K_e$, there exists a Lyapunov matrix $P \succ 0$ and a scalar $\alpha_L \in (0, 1)$ satisfying the Lyapunov inequality $(A_{\text{nom}} + B_{\text{nom}}K_e)^T P (A_{\text{nom}} + B_{\text{nom}}K_e) - P \le -\alpha_L P$. This inequality ensures that for the nominal, undisturbed error dynamics, the Lyapunov function $V(e) = e^T P e$ decays at a geometric rate.

The criticality of the realization $\zeta$ is then quantified by its ability to counteract this stabilizing decay. We define the one-step Lyapunov Violation Index as the amount by which the Lyapunov function increases, or fails to decrease as expected, at each step $k$:
\begin{equation}
    \mathcal{L}_k(\zeta) = V(e_{k+1}) - (1 - \alpha_L)V(e_k),
    \label{eq:lyapunov_violation_index}
\end{equation}
where the error trajectory $\{e_k\}$ (with $e_0 = 0$) is propagated forward using the full error dynamics \eqref{eq:error_dynamics} under the specific uncertainty realization $\zeta$ and the nominal plan $(z_{k|t}, v_{k|t})$ from the previous MPC solution. A positive value of $\mathcal{L}_k(\zeta)$ indicates a momentary growth in the error energy that exceeds the system's inherent stabilizing capability. The overall criticality of the entire realization, $\gamma(\zeta)$, is then conservatively defined as the maximum violation observed over the prediction horizon:
\begin{equation}
    \gamma(\zeta) = \max_{k \in \{0, \dots, N-1\}} \mathcal{L}_k(\zeta).
    \label{eq:criticality_metric}
\end{equation}
This metric provides a comprehensive measure of destabilizing potential, as it captures not only the immediate impact of an uncertainty but also its propagated effects through the system dynamics. Scenarios with a high $\gamma(\zeta)$ value are those most likely to drive the system state towards its constraint boundaries.

\paragraph{Gaussian Process Surrogate Modeling for Criticality Estimation.}
Directly evaluating the criticality function $\gamma(\zeta)$ requires an N-step simulation for each candidate uncertainty, rendering an exhaustive search computationally intractable. To overcome this limitation, we construct a computationally efficient surrogate model of the criticality function using Gaussian Process (GP) regression.

A prerequisite for effective GP modeling is the transformation of the variable-length, high-dimensional uncertainty realization $\zeta$ into a fixed-dimensional feature vector. We define a feature extraction operator $\Phi: \mathcal{U} \to \mathbb{R}^{d_f}$ that maps $\zeta$ to a feature vector $F = \Phi(\zeta)$. This process is crucial for capturing the essential characteristics that influence criticality. The feature vector $F$ is a concatenation of several components, including temporal features (e.g., a flattened window of the initial disturbance sequence), statistical features (e.g., moments of the disturbance), and spectral features (e.g., dominant frequencies from a Fourier Transform).

The GP then learns the mapping $\hat{\gamma}: \mathbb{R}^{d_f} \to \mathbb{R}$ from the feature space to the criticality value. It places a zero-mean Gaussian prior over the function space, with a covariance defined by a kernel function. We employ a squared exponential kernel with Automatic Relevance Determination (ARD):
\begin{equation}
    k(F_i, F_j) = \sigma_f^2 \exp\left(-\frac{1}{2} \sum_{d=1}^{d_f} \left(\frac{F_{i,d} - F_{j,d}}{l_d}\right)^2\right),
    \label{eq:gp_kernel}
\end{equation}
where the signal variance $\sigma_f^2$ and the length-scales $\{l_d\}$ are hyperparameters. Given a training dataset $\mathcal{D} = \{(F^{(i)}, \gamma^{(i)})\}_{i=1}^{n_{\text{train}}}$, where each $\gamma^{(i)}$ is the true criticality computed via \eqref{eq:criticality_metric}, the GP provides a full posterior predictive distribution for any new feature vector $F_*$. This distribution is Gaussian, $\mathcal{N}(\mu_{\text{GP}}(F_*), \sigma^2_{\text{GP}}(F_*))$, with the predictive mean and variance given by:
\begin{align}
    \mu_{\text{GP}}(F_*) &= \mathbf{k}_*^T (\mathbf{K} + \sigma_n^2 I)^{-1} \mathbf{y}, \label{eq:gp_mean_final} \\
    \sigma^2_{\text{GP}}(F_*) &= k(F_*, F_*) - \mathbf{k}_*^T (\mathbf{K} + \sigma_n^2 I)^{-1} \mathbf{k}_*, \label{eq:gp_variance_final}
\end{align}
where $\mathbf{K}$ is the $n_{\text{train}} \times n_{\text{train}}$ kernel matrix with entries $[\mathbf{K}]_{ij} = k(F^{(i)}, F^{(j)})$, $\mathbf{k}_*$ is the $n_{\text{train}} \times 1$ vector of kernel evaluations between $F_*$ and the training inputs, $\mathbf{y}$ is the $n_{\text{train}} \times 1$ vector of training targets $\{\gamma^{(i)}\}$, and $\sigma_n^2$ is the noise variance hyperparameter. The set of hyperparameters $\theta = \{\sigma_f^2, \{l_d\}_{d=1}^{d_f}, \sigma_n^2\}$ is optimized by maximizing the log marginal likelihood on the training data $\mathcal{D}$. This procedure automatically trades off model fit and complexity. The predictive variance $\sigma^2_{\text{GP}}(F_*)$ is particularly valuable, as it provides a principled measure of the model's confidence in its own prediction, which is the key enabler for the subsequent active learning step.

\paragraph{Active Discovery of Critical Scenarios.}
Armed with the fast-to-evaluate GP surrogate and its uncertainty estimates, we can efficiently search the space of possible uncertainties to find the most critical scenarios. We employ an active learning strategy based on the Upper Confidence Bound (UCB) acquisition function to intelligently balance exploitation (investigating regions the GP predicts to be highly critical) and exploration (investigating regions where the GP is uncertain).

The process begins by generating a large pool of $N_{\text{cand}}$ candidate uncertainty realizations. For each candidate $\zeta^{(i)}$, we compute its feature vector $F^{(i)} = \Phi(\zeta^{(i)})$ and then evaluate the GP surrogate to obtain the predictive mean $\mu_{\text{GP}}(F^{(i)})$ and variance $\sigma^2_{\text{GP}}(F^{(i)})$. We then compute its Pessimistic Criticality Estimate (PCE):
\begin{equation}
    \text{PCE}(\zeta^{(i)}) = \mu_{\text{GP}}(F^{(i)}) + \kappa_{\text{ucb}} \sigma_{\text{GP}}(F^{(i)}),
    \label{eq:pce_ucb}
\end{equation}
where $\kappa_{\text{ucb}} > 0$ is a tunable parameter. We rank all candidates according to their PCE values and select the top $K_{\text{crit}}$ scenarios to form the critical scenario set, $\mathcal{C}_{\text{crit}}$. This targeted discovery process is significantly more sample-efficient at finding high-impact, low-probability events than unstructured sampling methods.

\paragraph{High-Fidelity Simulation and LPES Construction.}
The final component of the loop translates the identified set of abstract critical scenarios, $\mathcal{C}_{\text{crit}}$, into a concrete, computationally tractable set representation for the MPC. For each critical scenario $\zeta^{(j)} \in \mathcal{C}_{\text{crit}}$, we perform a full N-step simulation of the error dynamics \eqref{eq:error_dynamics} to obtain the precise error trajectory $\{e^{(j)}_{k|t}\}_{k=0}^N$.

These simulated trajectories are aggregated to construct the time-varying Learned Prediction-Error Sets. For each prediction step $k \in \{0, \dots, N-1\}$, the LPES, $\mathcal{S}_{k|t}$, is formally defined as the convex hull of the simulated error state endpoints at that step:
\begin{equation}
    \mathcal{S}_{k|t} = \text{Conv}\left(\{e^{(j)}_{k|t} \mid j=1, \dots, K_{\text{crit}}\}\right).
    \label{eq:lpes_conv}
\end{equation}
To maintain computational tractability within the online MPC, we employ an axis-aligned outer approximation of this set, characterized by support vectors $s_{k|t} \in \mathbb{R}^{n_x}_{\ge 0}$. The $i$-th component of the support vector is computed as the maximum absolute value observed along that dimension among all critical error trajectories:
\begin{equation}
    [s_{k|t}]_i = \max_{j \in \{1, \dots, K_{\text{crit}}\}} |[e^{(j)}_{k|t}]_i|, \quad \text{for } i=1, \dots, n_x.
    \label{eq:lpes_support_vector}
\end{equation}
This set of support vectors, $\{s_{k|t}\}_{k=0}^{N-1}$, is the final output of the medium-frequency loop and is subsequently used to define the tightening in the robust MPC problem \eqref{eq:robust_mpc_qp}.

\begin{figure}
    \centering
    \includegraphics[width=1\linewidth]{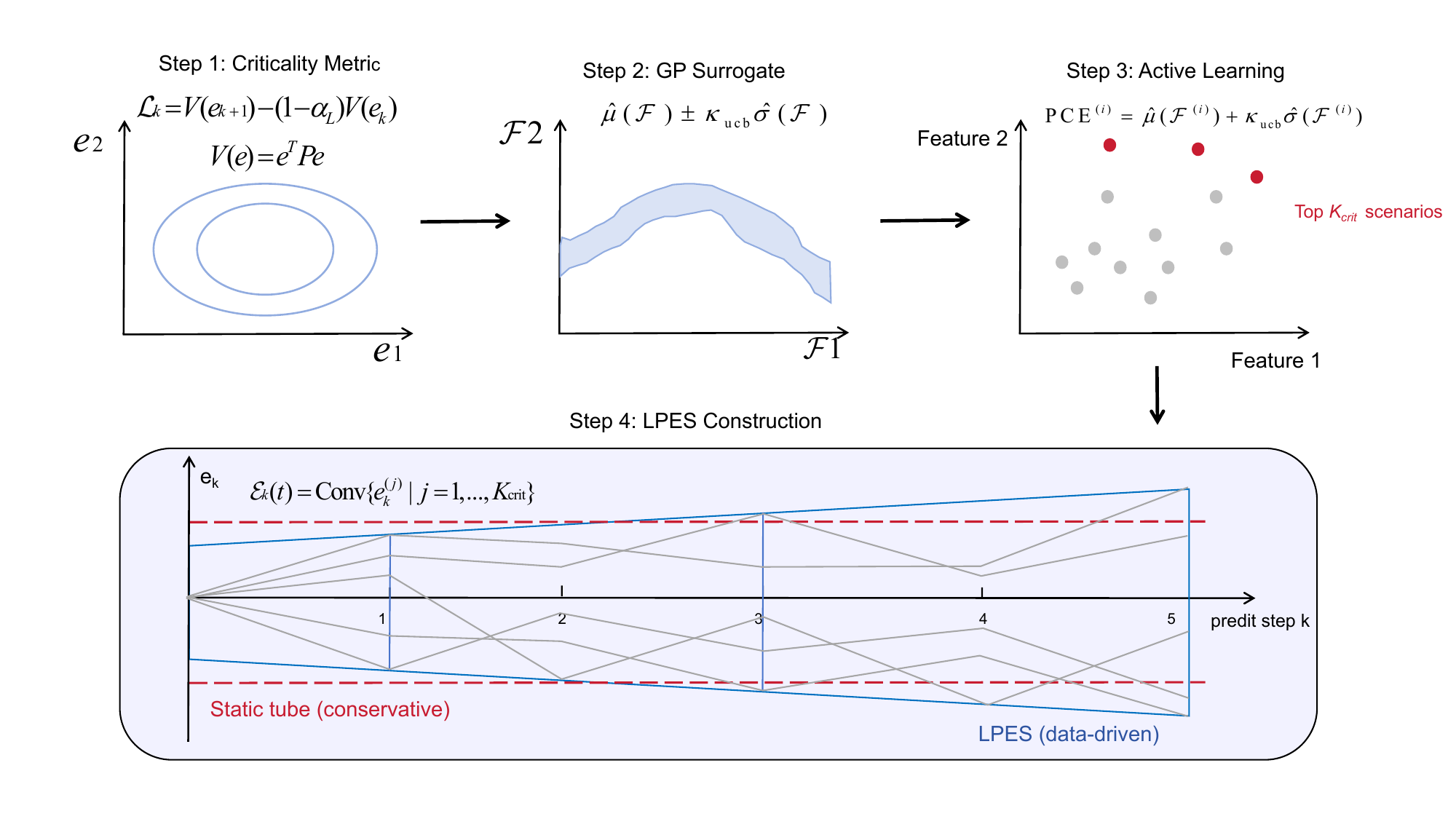}
    \caption{The medium-frequency loop for Learned Prediction-Error Set (LPES) construction.}
    \label{fig:lpes_construction}
\end{figure}

\subsubsection{Low-Frequency Loop: Self-Correcting Risk Regulation}
\label{sec:low_freq_loop}

While the medium-frequency loop proactively characterizes structured uncertainties via the LPES, a mechanism is still required to compensate for any residual model mismatch, unstructured disturbances, or inadequacies of the learned error set. Furthermore, a method is needed to ensure that the long-term empirical rate of constraint violation converges precisely to the user-specified risk level $\delta$. The low-frequency risk regulation loop, which updates the adaptive safety margin $\beta_t$, is designed for this purpose. It functions as a reactive, self-correcting outer loop that provides the ultimate guarantee on chance constraint satisfaction.

\paragraph{The Challenge of Learning from Rare Events.}
A straightforward approach to tune $\beta_t$ would be to employ a standard stochastic approximation (SA) scheme. Such a scheme would increase $\beta_t$ upon observing a physical constraint violation and decrease it otherwise, aiming to drive the violation probability to $\delta$. This can be formulated as:
\begin{equation}
    \beta_{t+1} = \Pi_{\mathcal{B}} \left[ \beta_t - \alpha_t \left( \mathbb{I}(h(x_t) > 0) - \delta \right) \right],
    \label{eq:naive_sa}
\end{equation}
where $h(x_t) > 0$ denotes a constraint violation, $\mathbb{I}(\cdot)$ is the indicator function, $\alpha_t$ is a learning rate, and $\Pi_{\mathcal{B}}$ is a projection onto a valid range for $\beta_t$, e.g., $[0, \beta_{\max}]$.

However, this naive approach is fundamentally flawed in the context of a high-performance robust control system. By design, the combined action of the robust MPC and the LPES makes physical constraint violations rare, low-probability events. Consequently, the learning signal $\mathbb{I}(h(x_t) > 0)$ is almost always zero. The SA algorithm is thus starved of corrective feedback, causing $\beta_t$ to perpetually decrease until the system's robustness margin is eroded, eventually leading to a cascade of violations without a reliable mechanism for recovery. This signal sparsity problem necessitates a more sophisticated learning architecture.

\paragraph{A Dynamic Target Compensation Framework.}
To overcome the challenge of learning from sparse signals, we propose a framework that decouples the learning trigger from the rare physical violation event. This is achieved by defining a more frequent "learning event" and then correcting for the statistical bias introduced by this redefinition.

First, we define a dynamic, internal safety boundary that is coupled with the current robustness level $\beta_t$. The learning boundary, $m_s(t)$, is defined as:
\begin{equation}
    m_s(t) = c_m \beta_t,
    \label{eq:learning_boundary}
\end{equation}
where $c_m > 0$ is a constant gain. This coupling creates a stabilizing negative feedback loop: as $\beta_t$ increases (making the system more robust), the boundary $m_s(t)$ also increases, making the learning trigger less sensitive. We then define the "Learning Event," $L_t$, as the event where the constraint function $h(x_t)$ exceeds the negative of this boundary:
\begin{equation}
    L_t \iff h(x_t) > -m_s(t).
    \label{eq:learning_event}
\end{equation}

\begin{figure}
    \centering
    \includegraphics[width=1\linewidth]{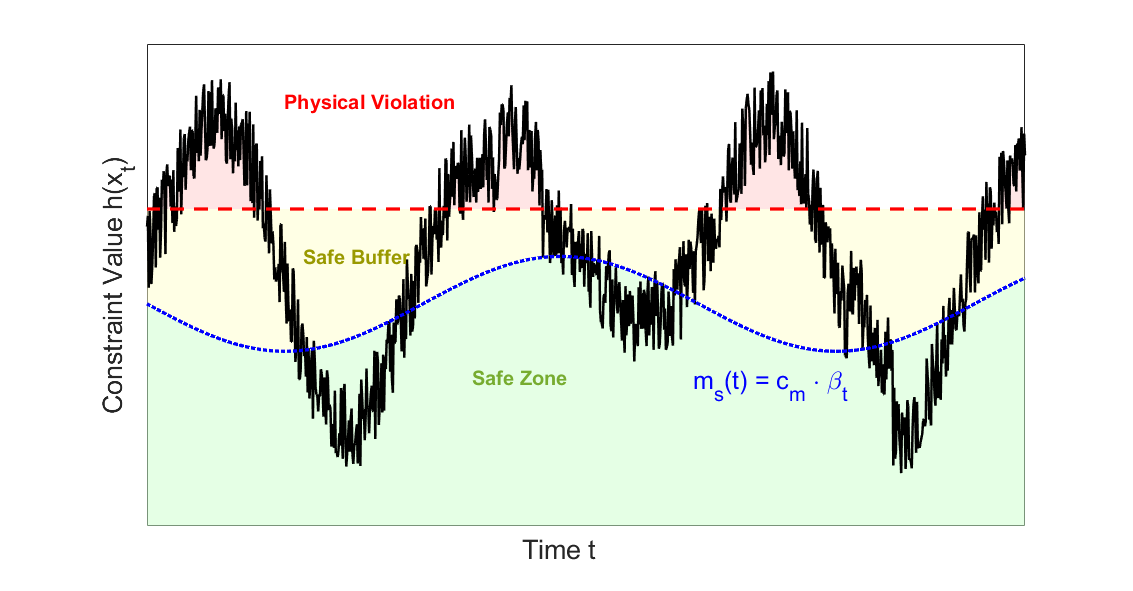}
    \caption{Conceptual illustration of the event regions used in the low-frequency risk regulation loop. To overcome the signal sparsity of physical constraint violations (red region), we define a dynamic learning boundary $m_s(t) = c_m \cdot \beta_t$ (blue dotted line). The "Learning Event" $\mathcal{L}_t$ is defined as the union of the "Physical Violation" region ($h(x_t) > 0$) and the "Safe Buffer" region ($-m_s(t) < h(x_t) \leq 0$). This creates a frequent and informative signal for the stochastic approximation scheme, while the statistical bias introduced by the "Safe Buffer" is actively compensated for.}
    \label{fig:enter-label}
\end{figure}

Since $-m_s(t)$ is typically a value within the "safe" region (i.e., less than zero), the learning event $L_t$ occurs far more frequently than the physical violation event $h(x_t) > 0$, providing a rich and persistent signal for the learning algorithm.

However, by shifting the learning trigger, we have also shifted the probabilistic target of the SA algorithm. The probability of a learning event, $\mathbb{P}(L_t)$, can be decomposed into the sum of two disjoint probabilities: the probability of a physical violation and the probability of the state residing in the "safe buffer" zone. This relationship is formalized in the following theorem.

\begin{theorem}[Probabilistic Decomposition and Inherent Bias]
\label{thm:prob_decomp}
The probability of a learning event, $\mathbb{P}(L_t)$, where $L_t$ is defined in \eqref{eq:learning_event}, can be decomposed as:
\begin{equation}
    \mathbb{P}(L_t) = \mathbb{P}(h(x_t) > 0) + \mathbb{P}(-m_s(t) < h(x_t) \le 0).
    \label{eq:prob_decomposition}
\end{equation}
\end{theorem}
\begin{proof}
The event $L_t = \{x_t \mid h(x_t) > -m_s(t)\}$ is the union of two disjoint events: the physical violation event $V_t = \{x_t \mid h(x_t) > 0\}$ and the safe buffer event $B_t = \{x_t \mid -m_s(t) < h(x_t) \le 0\}$. Since $V_t \cap B_t = \emptyset$, the result follows directly from the additivity axiom of probability.
\end{proof}

Theorem~\ref{thm:prob_decomp} reveals that if an SA algorithm drives $\mathbb{P}(L_t) \to \delta$, the physical violation probability will converge to $\mathbb{P}(h(x_t)>0) \to \delta - \mathbb{P}(B_t)$, resulting in a systematic and undesirable conservatism. To counteract this bias, we introduce a Dynamic Target Compensation mechanism. Instead of tracking the static target $\delta$, the algorithm is designed to track a time-varying learning target, $\delta_L(t)$, which actively accounts for the probability of the state residing in the safe buffer zone:
\begin{equation}
    \delta_L(t) = \delta + \hat{\mathbb{P}}_t(B_t),
    \label{eq:dynamic_target}
\end{equation}
where $\hat{\mathbb{P}}_t(B_t)$ is an online estimate of the buffer probability $\mathbb{P}(-m_s(t) < h(x_t) \le 0)$. This estimate is computed empirically from a sliding window of $W$ recent constraint function values, $\{h(x_i)\}_{i=t-W+1}^t$:
\begin{equation}
    \hat{\mathbb{P}}_t(B_t) \approx \frac{1}{W} \sum_{i=t-W+1}^{t} \mathbb{I}(-m_s(i) < h(x_i) \le 0).
    \label{eq:buffer_prob_estimate}
\end{equation}

\paragraph{Self-Correcting Stochastic Approximation Scheme}
With the components for bias compensation in place, we can now formulate the final, refined stochastic approximation scheme for updating the adaptive margin $\beta_t$. The stochastic error term for the update, $e_{\text{SA}}(t)$, is computed with respect to the compensated learning target $\delta_L(t)$:
\begin{equation}
    e_{\text{SA}}(t) = \mathbb{I}(h(x_t) > -m_s(t)) - \delta_L(t).
\end{equation}
The final update law for the adaptive robustness margin is then given by:
\begin{equation}
    \beta_{t+1} = \Pi_{\mathcal{B}} \left[ \beta_t - \alpha_t e_{\text{SA}}(t) - \gamma_t(\beta_t - \bar{\beta}) \right],
    \label{eq:final_beta_update}
\end{equation}
where $\Pi_{\mathcal{B}}$ is the projection onto the valid interval $[0, \beta_{\max}]$, $\alpha_t$ is the primary learning rate, and the final term, $-\gamma_t(\beta_t - \bar{\beta})$, is a mean-reversion component with a small rate $\gamma_t \ll \alpha_t$. This term provides additional stability to the learning process by gently pulling $\beta_t$ towards a pre-defined baseline margin $\bar{\beta}$, preventing unconstrained drift.

This complete outer-loop mechanism establishes a robust, dual-feedback system. A fast inner loop, driven by the learning rate $\alpha_t$, tracks the dynamic target $\delta_L(t)$, while a slower outer loop, implemented via the estimation in \eqref{eq:buffer_prob_estimate}, corrects the target itself. This architecture ensures that the closed-loop system robustly and accurately converges to the desired physical risk level $\delta$, achieving a near-optimal balance between performance and safety.

\begin{figure}
    \centering
    \includegraphics[width=1\linewidth]{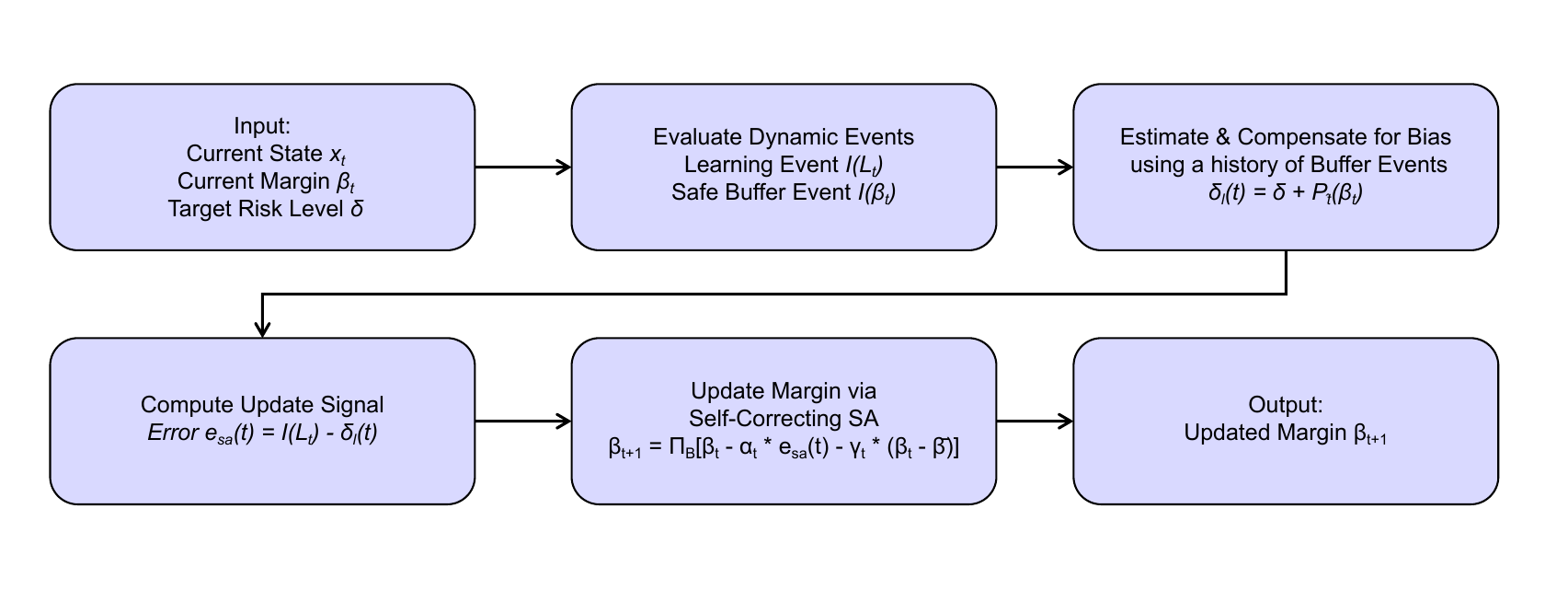}
    \caption{Flowchart of the low-frequency, self-correcting risk regulation loop for updating the adaptive safety margin $\beta_t$.}
    \label{fig:enter-label}
\end{figure}

\subsection{Algorithm Summary}
\label{sec:algorithm_summary}

The components described in the preceding sections—the robust inner-loop controller, the medium-frequency risk assessment engine, and the low-frequency risk regulation loop—are integrated into a single, cohesive algorithm. The complete operational flow of the RAAR-MPC framework is summarized in Algorithm~\ref{alg:raar-mpc}.

\begin{algorithm}[H]
\caption{The RAAR-MPC Framework (Concise)}
\label{alg:raar-mpc}
\begin{algorithmic}[1] 
\State \textbf{Initialize:} System  controller parameters $(A_{\text{nom}}, B_{\text{nom}}, Q, R, P, \mathcal{X}, \mathcal{U}, K_e, N)$.
\State \textbf{Initialize:} Learning parameters $(\delta, \beta_0, \bar{\beta}, \{\alpha_t, \gamma_t\}, W, M, K_{\text{crit}})$.
\State \textbf{Initialize:} GP model $\mathcal{GP}$, LPES $\mathcal{S}(0)$, history buffer $\mathcal{H} \leftarrow \emptyset$.

\For{$t = 0, 1, 2, \dots$}
    \State Measure current state $x_t$.
    \State Solve robust QP (8) to find nominal plan $(z_t^*, v_t^*)$.
    \State Apply control $u_t \leftarrow v_{0|t}^* + K_e(x_t - z_{0|t}^*)$.
    
    \Statex \Comment{Low-frequency risk regulation loop}
    \State Update adaptive safety margin $\beta_{t+1}$ using self-correcting stochastic approximation (Eq. 23-25).
    
    \Statex \Comment{Medium-frequency risk assessment loop}
    \If{$t \pmod M == 0$}
        \State Update LPES $\mathcal{S}(t+1)$ via risk assessment engine (Eq. 15-17).
    \Else
        \State $\mathcal{S}(t+1) \leftarrow \mathcal{S}(t)$.
    \EndIf
\EndFor
\end{algorithmic}
\end{algorithm}

The algorithm proceeds at each time step $t$ by first solving the robust MPC problem \eqref{eq:robust_mpc_qp} using the current LPES support vectors $s_{k|t}$ and adaptive margin $\beta_t$. The first element of the resulting optimal control sequence is then applied to the system. Concurrently, the low-frequency loop updates the adaptive margin $\beta_t$ based on the observed constraint behavior. Periodically, every $M$ steps, the medium-frequency loop is triggered to update the GP surrogate model with new data and subsequently re-compute a new set of LPES support vectors $\{s_{k|t+1}\}$ for the next operational phase. This multi-timescale orchestration ensures that the computationally intensive learning tasks do not interfere with the real-time control execution.

\begin{figure}
    \centering
    \includegraphics[width=1\linewidth]{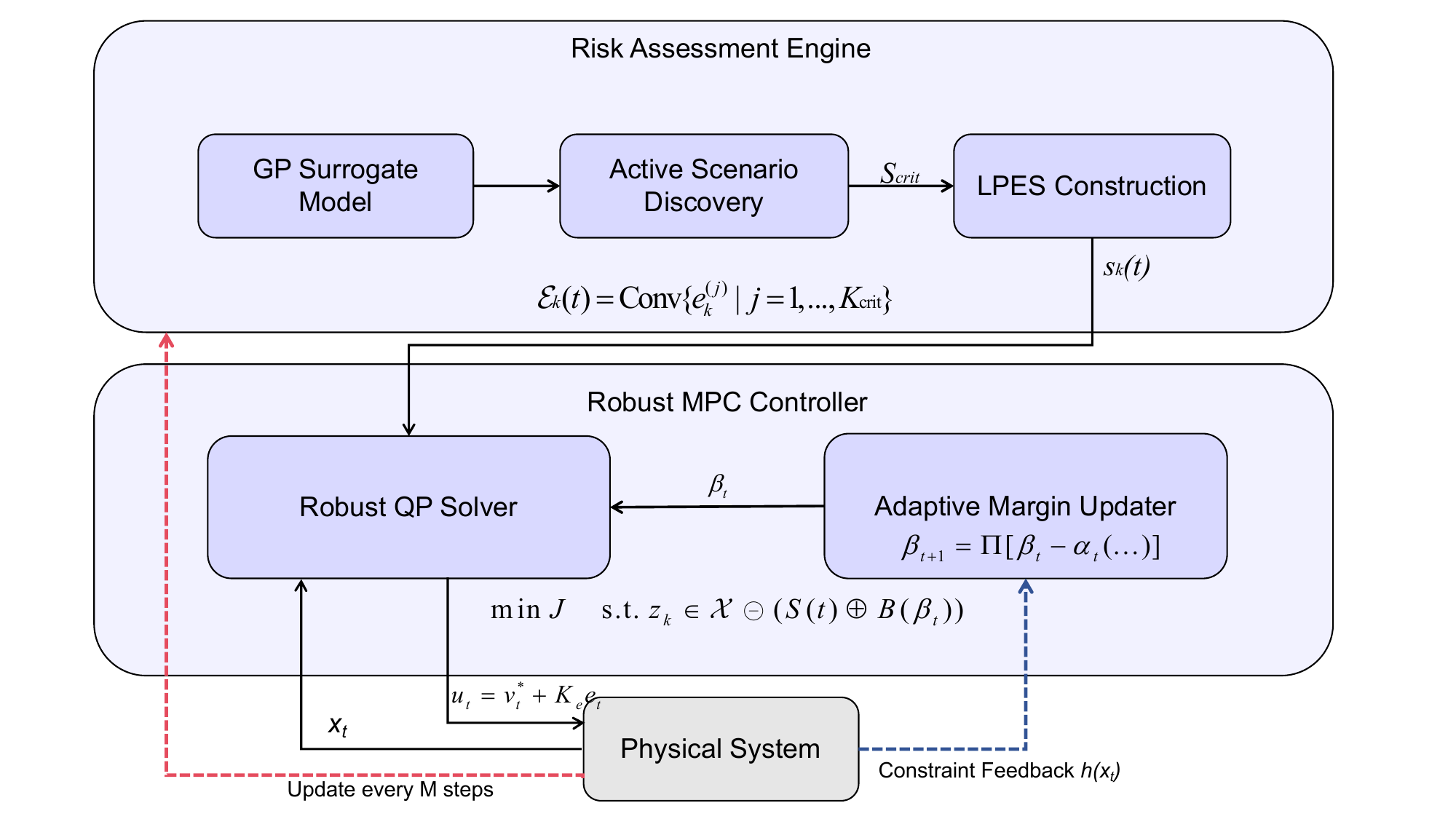}
    \caption{The proposed RAAR-MPC framework architecture.}
    \label{fig:enter-label}
\end{figure}

\section{THEORETICAL ANALYSIS}
\label{sec:theoretical_analysis}
In this section, we provide a formal analysis of the proposed Risk-Aware Adaptive Robust MPC (RAAR-MPC) framework. Our objective is to rigorously establish its key theoretical properties, namely recursive feasibility, closed-loop stability, and the satisfaction of probabilistic constraints. Our theoretical development relies on the following set of standard and problem-specific assumptions.

\subsection{Assumptions}
\begin{assumption}
 The state constraint set $\mathcal{X} \subset \mathbb{R}^{n_x}$ and the input constraint set $\mathcal{U} \subset \mathbb{R}^{n_u}$ are compact and contain the origin in their respective interiors. The set of system uncertainties, $\Delta$, and the disturbance support, $\mathcal{D}$, are compact.
\end{assumption}
\begin{assumption}
 There exists a state feedback gain matrix $K_e \in \mathbb{R}^{n_u \times n_x}$ for the error dynamics and a common Lyapunov matrix $P_e > 0$ such that for some scalar $\alpha_e \in (0,1)$, the following discrete-time Lyapunov inequality holds for all possible uncertainty realizations $(A_k, B_k) \in \Delta$:
\begin{equation}
(A_k + B_k K_e)^T P_e (A_k + B_k K_e) - P_e \leq -\alpha_e P_e.
\end{equation}  
\end{assumption}
\begin{assumption}
    The terminal set $\mathcal{X}_f \subseteq \mathcal{X}$, the terminal cost weighting matrix $P_f > 0$, and a terminal feedback gain $K_f$ are chosen to satisfy standard terminal conditions for robust MPC.
\end{assumption}
\begin{assumption}
 The Gaussian Process (GP) surrogate model for the criticality function $\gamma(\zeta)$ satisfies the following properties:
\begin{enumerate}
    \item[i)] \textit{RKHS Assumption:} The true, unknown criticality function $\gamma(\cdot)$ belongs to the Reproducing Kernel Hilbert Space (RKHS) $\mathcal{H}_k$ generated by the kernel $k(\cdot, \cdot)$. The kernel is continuous and bounded on the compact uncertainty space $\mathcal{U}$.
    \item[ii)] \textit{High-Probability Confidence Bounds:} For a training dataset $\mathcal{D}_t$ of size $t$, let $\mu_t(\zeta)$ and $\sigma_t^2(\zeta)$ be the posterior mean and variance. For any confidence level $\delta_{GP} \in (0, 1)$, there exists a parameter $\beta_t$ such that with probability at least $1 - \delta_{GP}$ over the randomness of the GP:
    \begin{equation}
        |\gamma(\zeta) - \mu_t(\zeta)| \leq \beta_t \sigma_t(\zeta), \quad \forall \zeta \in \mathcal{U}.
    \end{equation}
    The parameter $\beta_t$ depends on $t$, the kernel $k$, and $\delta_{GP}$, and its value can be rigorously derived from information-theoretic results in GP optimization literature [30].
\end{enumerate} 
\end{assumption}

\textit{Remark 1.} This revised Assumption 4 replaces the original vague "well-calibrated" condition with a mathematically precise statement about the validity of the GP's confidence intervals over the entire uncertainty space. Part (ii) is the cornerstone for the subsequent analysis of the UCB-based active learning scheme.

\subsection{Learning-Based Uncertainty Characterization}

A cornerstone of the RAAR-MPC framework is its medium-frequency risk assessment engine, which actively learns a tight, data-driven characterization of the prediction error. This process identifies a set of $K_{crit}$ most critical uncertainty scenarios $\mathcal{C}_{crit}(t)$, which are then used to construct the Learned Prediction-Error Set (LPES), $\mathcal{S}_k(t)$. The following key lemma establishes a formal probabilistic guarantee on the coverage property of this set.

\begin{lemma}[Probabilistic Coverage of the LPES]
\label{lemma:lpes_coverage}
Under Assumption 4, for any desired risk level $\epsilon > 0$ and confidence level $1 - \delta_c > 0$, there exist minimal sample sizes $N_{cand}(\epsilon, \delta_c)$ and a minimal training data size $t_{min}(\epsilon, \delta_c)$ such that when the number of candidate scenarios $|\mathcal{U}_{cand}| \geq N_{cand}$ and the GP training data size $t \geq t_{min}$, the Learned Prediction-Error Set (LPES) constructed at learning cycle $t$ satisfies:
\begin{equation}
    \mathbb{P}_{\zeta_{new} \sim P(\mathcal{U})} \left( \bigcap_{k=0}^{N-1} \left\{ e_k(\zeta_{new}) \in \mathcal{S}_k(t) \right\} \right) \geq 1 - \epsilon,
\end{equation}
where $P(\mathcal{U})$ is the true distribution of the uncertainty. This statement holds with probability at least $1 - \delta_c$ over the randomness of the GP model and the sampling of $\mathcal{U}_{cand}$.
\end{lemma}

\begin{proof}
The proof proceeds by first establishing a formal link between the criticality measure $\gamma(\zeta)$ and the prediction error norm, then showing that a coverage failure implies the discovery of a more critical scenario than those already known, and finally bounding the probability of such a discovery using GP confidence bounds.

We formalize the intuition that a larger error corresponds to a higher criticality value.
\begin{proposition}
\label{prop:gamma_error_relation}
There exist constants $c_1, c_2 > 0$ and a strictly increasing function $h: \mathbb{R}_+ \to \mathbb{R}_+$ such that for all $\zeta \in \mathcal{U}$:
\begin{equation}
    c_1 \gamma(\zeta) \leq \max_{k \in \{0,..,N-1\}} \|e_k(\zeta)\|^2 \leq c_2 \gamma(\zeta).
\end{equation}
This implies the existence of a strictly increasing function $h$ such that $\max_k \|e_k(\zeta)\| \leq h(\gamma(\zeta))$.
\end{proposition}
\begin{proof}[Proof of Proposition \ref{prop:gamma_error_relation}]
The Lyapunov function $V(e) = e^T P_e e$ is norm-equivalent to $\|e\|^2$, i.e., $\lambda_{min}(P_e)\|e\|^2 \leq V(e) \leq \lambda_{max}(P_e)\|e\|^2$. The error dynamics (Eq. 6) show that $e_{k+1}$ is an affine function of $e_k$ and $\zeta$. By recursion, $e_k$ is an affine function of the initial error (zero) and the uncertainty sequence up to step $k-1$. Thus, $\max_k \|e_k(\zeta)\|$ is a continuous function of $\zeta$. Similarly, $\gamma(\zeta) = \max_k \{V(e_{k+1}(\zeta)) - (1-\alpha_e)V(e_k(\zeta))\}$ is a continuous function of $\zeta$. On the compact set $\mathcal{U}$, both functions are bounded. The norm-equivalence of $V(e)$ and $\|e\|^2$ and the structure of the error dynamics establish the existence of such constants $c_1$ and $c_2$. The existence of $h$ follows directly.
\end{proof}

Let $E_{fail}$ be the event that a new scenario $\zeta_{new}$ causes a coverage failure.
\begin{proposition}[Coverage Failure Implication]
\label{prop:failure_implication}
The event $E_{fail} = \{\zeta_{new} | \exists k: e_k(\zeta_{new}) \notin \mathcal{S}_k(t)\}$ implies the event $B = \{\zeta_{new} | \gamma(\zeta_{new}) > \max_{\zeta^{(j)} \in \mathcal{C}_{crit}(t)} \gamma(\zeta^{(j)})\}$.
\end{proposition}
\begin{proof}[Proof of Proposition \ref{prop:failure_implication}]
The set $\mathcal{S}_k(t)$ is the axis-aligned bounding box of the error vectors from $\mathcal{C}_{crit}(t)$, i.e., its support vector is $[s_{k|t}]_i = \max_j |[e_k(\zeta^{(j)})]_i|$. The condition $e_k(\zeta_{new}) \notin \mathcal{S}_k(t)$ implies that there exists at least one dimension $i$ such that $|[e_k(\zeta_{new})]_i| > [s_{k|t}]_i$. This directly leads to $\max_k \|e_k(\zeta_{new})\| > \max_j (\max_k \|e_k(\zeta^{(j)})\|)$. Applying Proposition \ref{prop:gamma_error_relation}, we have:
\begin{equation}
\begin{split}  
   \gamma(\zeta_{new}) &\geq \frac{1}{c_2} \max_k \|e_k(\zeta_{new})\|^2 \\&> \frac{1}{c_2} \max_j (\max_k \|e_k(\zeta^{(j)})\|^2) \\&\geq \frac{c_1}{c_2} \max_j \gamma(\zeta^{(j)})
\end{split}
\end{equation}
Since $c_1/c_2 > 0$, this establishes the implication $E_{fail} \Rightarrow B$.
\end{proof}

Let $\gamma^* = \sup_{\zeta \in \mathcal{U}} \gamma(\zeta)$. We bound the probability of event $B$ by analyzing the performance of the UCB selection strategy. Let $\mathcal{E}_{opt}(\varepsilon_{opt}) = \{\max_{j} \gamma(\zeta^{(j)}) < \gamma^* - \varepsilon_{opt}\}$ be the event that the UCB algorithm fails to find a scenario with near-optimal criticality.
The probability of event B can be bounded using the law of total probability:
\begin{equation}
\begin{split}
    \mathbb{P}(B) &= \mathbb{P}(B|\mathcal{E}_{opt}) \mathbb{P}(\mathcal{E}_{opt}) + \mathbb{P}(B|\mathcal{E}_{opt}^c) \mathbb{P}(\mathcal{E}_{opt}^c)\\& \leq \mathbb{P}(\mathcal{E}_{opt}) + \mathbb{P}(B|\mathcal{E}_{opt}^c)
\end{split}
\end{equation}
The term $\mathbb{P}(\mathcal{E}_{opt})$ is the sub-optimality gap of the GP-UCB algorithm, which can be bounded. For a sufficiently large candidate set size $N_{cand}$ and training data size $t$, GP optimization theory guarantees that this probability is small. Specifically, it can be bounded by a term that decreases with $N_{cand}$ (controlling the discretization error of $\mathcal{U}$) and $t$ (improving the GP model accuracy) [30]. Let this bound be $\mathbb{P}(\mathcal{E}_{opt}) \leq \delta_{opt}(N_{cand}, t)$.

The second term is $\mathbb{P}(B|\mathcal{E}_{opt}^c) \leq \mathbb{P}(\gamma(\zeta_{new}) > \gamma^* - \varepsilon_{opt})$. Let $F_\gamma$ be the cumulative distribution function (CDF) of the random variable $\gamma(\zeta_{new})$. This probability is $1 - F_\gamma(\gamma^* - \varepsilon_{opt})$. Since $\gamma$ is a continuous function on a compact set, its distribution is absolutely continuous. Therefore, for any $\epsilon' > 0$, we can choose $\varepsilon_{opt}$ small enough such that $1 - F_\gamma(\gamma^* - \varepsilon_{opt}) < \epsilon'$.

Combining these bounds, the total probability of coverage failure is bounded by:
\begin{equation}
    \mathbb{P}(E_{fail}) \leq \mathbb{P}(B) \leq \delta_{opt}(N_{cand}, t) + (1 - F_\gamma(\gamma^* - \varepsilon_{opt})).
\end{equation}
To achieve a final risk of $\epsilon$, we can allocate the error budget. For instance, choose $\varepsilon_{opt}$ such that $1 - F_\gamma(\gamma^* - \varepsilon_{opt}) < \epsilon/2$. Then, choose $N_{cand}$ and $t_{min}$ large enough such that for $t \ge t_{min}$, we have $\delta_{opt}(N_{cand}, t) < \epsilon/2$. This ensures $\mathbb{P}(E_{fail}) < \epsilon$. The entire argument holds with probability at least $1-\delta_{GP}$ from Assumption 4, so we can set $\delta_c = \delta_{GP}$ or incorporate it into the budget. This completes the proof of Lemma \ref{lemma:lpes_coverage}.
\end{proof}

\textit{Remark 2.} This lemma is of paramount importance as it formally connects the output of our data-driven risk assessment engine to the requirements of the robust MPC controller. It replaces any flawed analogy to scenario theory with a rigorous argument based on the properties of Gaussian Processes and UCB-based active learning. It provides the quantifiable guarantee that allows us to reason about the probabilistic satisfaction of state and input constraints for the true system.

\subsection{Recursive Feasibility}

The recursive feasibility of the proposed MPC scheme ensures that the optimization problem is solvable at every time step. We establish this property by showing that a feasible solution at an arbitrary time step $t$ can be used to construct a feasible candidate solution for the subsequent time step $t+1$.

\begin{theorem}[Recursive Feasibility of RAAR-MPC]
\label{thm:recursive_feasibility}
Let Assumptions 1-4 hold. Suppose at time $t$, the RAAR-MPC optimization problem (8) is feasible for a measured state $x_t$. Then, the following statements hold:

\begin{enumerate}
    \item[a)] Robust Recursive Feasibility: The optimization problem (8) remains feasible for the state $x_{t+1}$ at the subsequent time step $t+1$ for all possible uncertainty realizations, potentially with a non-zero slack variable $\varepsilon_{t+1} \ge 0$.
    
    \item[b)] High-Probability Recursive Feasibility: Let $E_{cover}(t)$ be the event that the one-step prediction error is contained within its learned set, i.e., $e_{1|t} \in \mathcal{S}_1(t)$. Conditioned on this event, which occurs with probability at least $1-\epsilon_c$ by Lemma \ref{lemma:lpes_coverage}, and under a mild condition on the adaptation rate (Assumption \ref{assump:adaptation_rate} below), the optimization problem (8) is feasible at time $t+1$ with a slack variable $\varepsilon_{t+1} = 0$.
\end{enumerate}
\end{theorem}

\begin{proof}
Let $z_t^* = \{z_{k|t}^*\}_{k=0}^N$ and $v_t^* = \{v_{k|t}^*\}_{k=0}^{N-1}$ be the optimal nominal sequences obtained by solving (8) at time $t$, with an optimal cost $J^*(x_t)$ and assuming $\varepsilon_t=0$. The true state at $t+1$ is $x_{t+1} = z_{1|t}^* + e_{1|t}$, where $e_{1|t}$ is the one-step prediction error. We construct a candidate solution $(\hat{z}_{t+1}, \hat{v}_{t+1})$ for the optimization problem at time $t+1$.

Candidate Solution Construction:
The candidate control sequence is formed by shifting the optimal plan from time $t$ and applying the terminal control law at the end of the horizon:
\begin{equation}
\label{eq:candidate_control}
    \hat{v}_{k|t+1} = 
    \begin{cases}
        v_{k+1|t}^* & \text{for } k = 0, \dots, N-2 \\
        K_f \hat{z}_{N-1|t+1} & \text{for } k = N-1
    \end{cases}
\end{equation}
The corresponding candidate nominal state sequence starts from the new measured state $\hat{z}_{0|t+1} = x_{t+1}$ and evolves according to the nominal dynamics:
\begin{equation}
\label{eq:candidate_state}
    \hat{z}_{k+1|t+1} = A_{nom} \hat{z}_{k|t+1} + B_{nom} \hat{v}_{k|t+1}.
\end{equation}
For $k=0, \dots, N-2$, this implies $\hat{z}_{k+1|t+1} = z_{k+2|t}^*$. The initial state is $\hat{z}_{0|t+1} = z_{1|t}^* + e_{1|t}$.

Proof of a) Robust Recursive Feasibility:
We must show that the candidate solution satisfies all constraints at $t+1$, possibly with $\varepsilon_{t+1} > 0$. The critical constraints are the state and input constraints for $k=0, \dots, N-1$.
From the construction, the candidate state can be expressed as $\hat{z}_{k|t+1} = z_{k+1|t}^* + A_{cl,e}^k e_{1|t}$, where $A_{cl,e}$ is the closed-loop error dynamics matrix associated with the ancillary controller $K_e$. At time $t$, feasibility implies $z_{k+1|t}^* \in \mathcal{X} \ominus \mathcal{U}_{total, k+1}(t)$. At time $t+1$, the constraint is $\hat{z}_{k|t+1} \in \mathcal{X} \ominus \mathcal{U}_{total, k}(t+1)$.
The potential violation of this constraint by the candidate solution is due to two factors: (1) the un-cancelled part of the propagated error $A_{cl,e}^k e_{1|t}$, and (2) the change in the tightening set from $\mathcal{U}_{total, k+1}(t)$ to $\mathcal{U}_{total, k}(t+1)$.
Under Assumption 1 (compactness), the one-step error $e_{1|t}$ is bounded. The update laws for $\beta_t$ (Eq. 25) and the construction of $\mathcal{S}(t)$ from bounded data ensure that the change in the set $\mathcal{U}_{total}$ is also bounded. Therefore, the magnitude of any potential constraint violation by the candidate solution is uniformly bounded. A sufficiently large but finite slack variable $\varepsilon_{t+1}$ can thus always be found to ensure feasibility. The satisfaction of the terminal constraint $\hat{z}_{N|t+1} \in \mathcal{X}_f$ follows from the definition of the candidate control and the properties of the terminal set (Assumption 3). This guarantees that the optimization problem is always solvable.

Proof of b) High-Probability Recursive Feasibility:
We now prove that under the high-probability event $E_{cover}(t)$, the candidate solution is feasible with $\varepsilon_{t+1}=0$. This requires a more detailed analysis of the constraint satisfaction. We introduce a mild assumption on the adaptation rate.
\begin{assumption}[Coherent Adaptation Rate]
\label{assump:adaptation_rate}
The LPES and adaptive margin are updated such that the uncertainty sets evolve coherently over one time step. Specifically, for $k=0, \dots, N-2$:
\begin{equation}
\label{eq:coherent_adaptation}
    \mathcal{R}_k(\mathcal{S}_1(t)) \oplus \mathcal{S}_k(t+1) \oplus B(\beta_{t+1}) \subseteq \mathcal{S}_{k+1}(t) \oplus B(\beta_t),
\end{equation}
where $\mathcal{R}_k(\mathcal{S}_1(t))$ is the $k$-step reachable set of the ancillary error dynamics starting from the set $\mathcal{S}_1(t)$.
\end{assumption}
\textit{Remark on Assumption \ref{assump:adaptation_rate}:} This assumption is reasonable. The ancillary controller $K_e$ is stabilizing (Assumption 2), so the reachable set $\mathcal{R}_k$ should contract. The LPES $\mathcal{S}(t)$ is typically updated only every $M \gg 1$ steps, so for most $t$, $\mathcal{S}(t+1)=\mathcal{S}(t)$. The adaptive margin $\beta_t$ evolves slowly due to a small learning rate. Thus, the condition essentially requires that the learned uncertainty does not exhibit pathological, fast-changing behavior, which is expected from a converged learning process.

Now we check the state constraint for the candidate solution $\hat{z}_{k|t+1}$ for $k=0, \dots, N-2$. We need to show that for any $w' \in \mathcal{U}_{total,k}(t+1)$, we have $\hat{z}_{k|t+1} + w' \in \mathcal{X}$.
The state of the candidate solution at step $k$ of the horizon at time $t+1$ is the sum of the nominal plan from the previous step and the evolution of the one-step error under the ancillary controller: $\hat{z}_{k|t+1} = z_{k+1|t}^* + \text{err}_k$, where $\text{err}_k \in \mathcal{R}_k(\{e_{1|t}\})$.

At time $t$, feasibility implied that for any $w \in \mathcal{U}_{total, k+1}(t)$:
\begin{equation}
    z_{k+1|t}^* + w \in \mathcal{X}.
\end{equation}
We are conditioned on the event $E_{cover}(t)$, which means $e_{1|t} \in \mathcal{S}_1(t)$. This implies that the propagated error $\text{err}_k$ is contained in the reachable set $\mathcal{R}_k(\mathcal{S}_1(t))$.
We need to check if $\hat{z}_{k|t+1} + w' \in \mathcal{X}$ for all $w' \in \mathcal{U}_{total, k}(t+1) = \mathcal{S}_k(t+1) \oplus B(\beta_{t+1})$. Let's expand the expression:
\begin{equation}
    \underbrace{z_{k+1|t}^*}_{\text{nominal plan}} + \underbrace{\text{err}_k}_{\text{propagated error}} + \underbrace{w'}_{\text{new uncertainty}} \in \mathcal{X}.
\end{equation}
We know $\text{err}_k + w' \in \mathcal{R}_k(\mathcal{S}_1(t)) \oplus \mathcal{S}_k(t+1) \oplus B(\beta_{t+1})$.
Using Assumption \ref{assump:adaptation_rate}, this sum is a subset of the uncertainty the previous plan was robust against:
\begin{equation}
    \text{err}_k + w' \in \mathcal{S}_{k+1}(t) \oplus B(\beta_t) = \mathcal{U}_{total, k+1}(t).
\end{equation}
Since $z_{k+1|t}^*$ is feasible with respect to any disturbance in $\mathcal{U}_{total, k+1}(t)$, it follows that the candidate state $\hat{z}_{k|t+1}$ satisfies its constraint without needing a slack variable. A similar argument holds for the input constraints, given the structure of the ancillary controller. The terminal constraint is satisfied as in part (a).
Thus, conditioned on the high-probability event $E_{cover}(t)$, the constructed candidate solution is feasible with $\varepsilon_{t+1}=0$.
\end{proof}


\subsection{Closed-Loop Stability and Convergence Analysis}

Having established recursive feasibility, we now analyze the long-term behavior of the closed-loop system. The objective is to prove that the system state $x_t$ and the adaptive robustness margin $\beta_t$ are jointly stable, and that $\beta_t$ converges to a vicinity of an ideal value that ensures the satisfaction of the chance constraint. This requires analyzing an augmented system comprising both the physical and the learning states.

To facilitate this analysis, we introduce the following assumption regarding the properties of the adaptive mechanism.

\begin{assumption}[Properties of the Adaptive Mechanism]
\label{assum:adaptive_mech}
\begin{enumerate}
    \item[\textit{i)}] (Bounded Parameters) The learning rates and the reference margin are bounded: $\alpha_t \le \alpha_{\max}$ and $\gamma_t \le \gamma_{\max}$ for all $t \ge 0$, and $\bar{\beta}$ is a known constant.
    \item[\textit{ii)}] (Monotonicity of Learning Probability) Let $P(L_t|\beta_t)$ denote the true probability of the learning event $L_t \equiv \{x | h(x) > -m_s(t)\}$, conditioned on a fixed margin $\beta_t$. This probability is assumed to be continuously differentiable and strictly decreasing in $\beta_t$, with a derivative uniformly bounded away from zero: $\frac{\partial P(L_t|\beta_t)}{\partial \beta_t} \le -c_p < 0$ for some constant $c_p > 0$.
    \item[\textit{iii)}] (Bounded Estimation Error) The error of the buffer probability estimator, defined as $\nu_t = \hat{P}_t(B_t) - P(B_t|\beta_t)$, where $P(B_t|\beta_t)$ is the true probability of the buffer event, has a bounded second moment: $\mathbb{E}[\nu_t^2] \le \nu_{\max}^2$.
\end{enumerate}
\end{assumption}

\textit{Remark 6:} Assumption \ref{assum:adaptive_mech} formalizes the well-posedness of the proposed self-correcting mechanism. Part \textit{i)} is a standard implementation choice. Part \textit{ii)} represents a key structural property, stating that increasing the robustness margin $\beta_t$ effectively and predictably reduces the frequency of learning events. This inherent negative feedback is essential for the stability of the adaptation process. Part \textit{iii)} posits that the sliding-window estimator for the buffer probability is sufficiently accurate, a condition supported by laws of large numbers for an appropriately chosen window size $W$.

Let $\beta^*$ denote the ideal, generally unknown, robustness margin that would make the physical constraint violation probability exactly equal to the target risk $\delta$, i.e., $P(h(x) > 0 | \beta = \beta^*) = \delta$. The stability analysis is centered on the following augmented Lyapunov function candidate:
\begin{equation}
    V(x_t, \beta_t) = J_t^*(x_t) + c_\beta (\beta_t - \beta^*)^2,
    \label{eq:lyap_augmented}
\end{equation}
where $J_t^*(x_t)$ is the optimal value of the MPC objective function at time $t$ for state $x_t$, and $c_\beta > 0$ is a positive weighting coefficient to be determined.

\begin{theorem}[Stochastic Stability and Convergence]
\label{thm:stochastic_stability}
Consider the system (1) under the RAAR-MPC control law with the update mechanism (25). Under Assumptions 1-5, for a sufficiently large choice of the weighting coefficient $c_\beta$, the closed-loop system is stochastically stable in the sense that the state $x_t$ and the adaptive margin $\beta_t$ are ultimately bounded in expectation. Specifically, there exist finite positive constants $B_x$ and $B_\beta$ such that:
\begin{equation}
    \limsup_{t \to \infty} \mathbb{E}[\|x_t\|^2] \le B_x \quad \text{and} \quad \limsup_{t \to \infty} \mathbb{E}[(\beta_t - \beta^*)^2] \le B_\beta
\end{equation}
\end{theorem}

\begin{proof}
The proof is based on analyzing the one-step expected drift of the Lyapunov function \eqref{eq:lyap_augmented}, conditioned on the filtration $\mathcal{F}_t$ which contains all information up to time $t$. We analyze the drift of each component of $V(x_t, \beta_t)$ separately.

1. Analysis of the Cost-to-Go Term $J_t^*(x_t)$:
Using a standard shifting argument from MPC theory, the optimal cost at time $t+1$ is upper-bounded by the cost of a feasible candidate solution constructed from the optimal solution at time $t$. This leads to the inequality:
\begin{equation}
    J_{t+1}^*(x_{t+1}) \le J_t^*(x_t) - \ell(z_{0|t}^*, v_{0|t}^*) + \Delta J_{\text{err},t} + \rho \varepsilon_{t+1}^{\text{cand}}
\end{equation}
Here, $\ell(z_{0|t}^*, v_{0|t}^*)$ is the stage cost at time $t$, which is positive definite with respect to the nominal state and input, i.e., $\ell(z_{0|t}^*, v_{0|t}^*) \ge c_z \|z_{0|t}^*\|^2$ for some $c_z > 0$. The term $\Delta J_{\text{err},t}$ represents the change in cost due to the prediction error $e_{1|t}$, and $\varepsilon_{t+1}^{\text{cand}}$ is the slack required for the candidate solution. Both terms can be upper-bounded by functions of $\|e_{1|t}\|$ and $\|e_{1|t}\|^2$. The error dynamics are driven by the nominal state and the disturbance, yielding $\mathbb{E}[\|e_{1|t}\|^2 | \mathcal{F}_t] \le \sigma_e \|x_t\|^2 + C_e$ for constants $\sigma_e, C_e > 0$. Taking the conditional expectation of the cost evolution, we obtain:
\begin{equation}
    \mathbb{E}[J_{t+1}^*(x_{t+1}) - J_t^*(x_t) | \mathcal{F}_t] \le -c_1 \|x_t\|^2 + C_J
    \label{eq:j_drift_final}
\end{equation}
where $c_1 > 0$ is a constant dependent on the MPC weighting matrices and system parameters, and $C_J$ is a constant dependent on the bound of the disturbance $d_k$.

2. Analysis of the Margin Term $(\beta_t - \beta^*)^2$:
We analyze the evolution of the squared error of the adaptive margin. The update law (25) is $\beta_{t+1} = \Pi_{\mathcal{B}} [ \beta_t - \alpha_t (\mathbb{I}(L_t) - \delta_l(t)) - \gamma_t(\beta_t - \bar{\beta}) ]$, where $\delta_l(t) = \delta + \hat{P}_t(B_t)$. The projection operator $\Pi_{\mathcal{B}}$ is non-expansive, thus its presence only decreases the squared error. We analyze the drift of the unprojected update for a worst-case bound:
\begin{align}
    (\beta_{t+1} - \beta^*)^2 &\le \left( (\beta_t - \beta^*) - \alpha_t(\mathbb{I}(L_t) - \delta_l(t)) - \gamma_t(\beta_t - \bar{\beta}) \right)^2 \nonumber \\
    \begin{split}
        &= (\beta_t - \beta^*)^2 \\
        &\quad - 2(\beta_t - \beta^*) \left( \alpha_t(\mathbb{I}(L_t) - \delta_l(t)) + \gamma_t(\beta_t - \bar{\beta}) \right) \\
        &\quad + R_t
    \end{split}
\end{align}
where $R_t$ collects all second-order terms, which are bounded since $\alpha_t, \gamma_t$ are bounded. Taking the conditional expectation yields the drift of the margin error term. Let $\tilde{\beta}_t = \beta_t - \beta^*$ denote the deviation from the ideal margin. The one-step drift can be expressed as:
\begin{align}
    \mathbb{E}[(\beta_{t+1} - \beta^*)^2 - (\beta_t - \beta^*)^2 | \mathcal{F}_t] \le  
    -2\alpha_t \tilde{\beta}_t \mathbb{E}[\mathbb{I}(L_t) - \delta_l(t) | \mathcal{F}_t] \nonumber
    \\-2\gamma_t \tilde{\beta}_t (\beta_t - \bar{\beta}) 
     + \bar{C}_{\beta}.
\end{align}
The core stochastic term, which represents the innovation of the learning process, can be analyzed as follows:
\begin{align}
    \mathbb{E}[\mathbb{I}(L_t) - \delta_l(t) | \mathcal{F}_t] =& P(L_t|\beta_t) - (\delta + P(B_t|\beta_t)) - \mathbb{E}[\nu_t|\mathcal{F}_t] \nonumber \nonumber\\
    =& \left(P(h(x)>0|\beta_t) + P(B_t|\beta_t)\right)\nonumber\\& - (\delta + P(B_t|\beta_t)) - \mathbb{E}[\nu_t|\mathcal{F}_t] \nonumber\nonumber \\
    =& P(h(x)>0|\beta_t) - P(h(x)>0|\beta^*)\nonumber \\&- \mathbb{E}[\nu_t|\mathcal{F}_t]
\end{align}
Using the Mean Value Theorem and Assumption \ref{assum:adaptive_mech}(ii), the systematic drift term is bounded: $P(h(x)>0|\beta_t) - P(h(x)>0|\beta^*) \le -c_p \tilde{\beta}_t$. Let $e_{\nu,t} = \mathbb{E}[\nu_t|\mathcal{F}_t]$ denote the conditional estimation error. Substituting these into the drift inequality gives:
\begin{align}
    \mathbb{E}[\Delta \tilde{\beta}_t^2 | \mathcal{F}_t] &\le -2\alpha_t \tilde{\beta}_t (-c_p \tilde{\beta}_t - e_{\nu,t})\nonumber\\& - 2\gamma_t \tilde{\beta}_t (\tilde{\beta}_t + \beta^* - \bar{\beta}) + \bar{C}_{\beta} \nonumber \\
    &= -2\alpha_t c_p \tilde{\beta}_t^2 + 2\alpha_t \tilde{\beta}_t e_{\nu,t} \nonumber\\&- 2\gamma_t \tilde{\beta}_t^2 - 2\gamma_t \tilde{\beta}_t (\beta^* - \bar{\beta}) + \bar{C}_{\beta}.
    \label{eq:beta_drift_intermediate}
\end{align}
Applying Young's inequality, $2ab \le \eta a^2 + b^2/\eta$, to the cross terms involving $|\beta_t - \beta^*|$:
\begin{equation}
    \mathbb{E}[\Delta (\beta_t - \beta^*)^2 | \mathcal{F}_t] \le -c_2' (\beta_t - \beta^*)^2 + C_{\beta},
    \label{eq:beta_drift_final}
\end{equation}
where $c_2' = 2\alpha_t c_p + 2\gamma_t - \eta_1 - \eta_2 > 0$ for sufficiently small positive $\eta_1, \eta_2$, and $C_\beta$ is a constant dependent on $\nu_{\max}^2$ and other parameters.

3. Combined Lyapunov Drift Analysis:
Combining the results from \eqref{eq:j_drift_final} and \eqref{eq:beta_drift_final}, the total drift of the augmented Lyapunov function $V(x_t, \beta_t)$ is:
\begin{align}
    \mathbb{E}[\Delta V_t | \mathcal{F}_t] \le& (-c_1 \|x_t\|^2 + C_J) \nonumber\\&+ c_\beta(-c_2'(\beta_t - \beta^*)^2 + C_\beta) \nonumber \\
    =& -c_1 \|x_t\|^2 - c_\beta c_2' (\beta_t - \beta^*)^2 \nonumber\\&+ (C_J + c_\beta C_\beta).
\end{align}
Let $c_2 = c_\beta c_2'$. The inequality can be written as:
\begin{equation}
    \mathbb{E}[\Delta V_t | \mathcal{F}_t] \le -c_1 \|x_t\|^2 - c_2 (\beta_t - \beta^*)^2 + C_{\text{total}}.
\end{equation}
This expression shows that for states $(x_t, \beta_t)$ outside a specific compact set around $(0, \beta^*)$, the expected one-step drift of the Lyapunov function is negative. By the theory of stochastic stability for discrete-time processes, this condition implies that the process is ultimately bounded in mean square. The size of the ultimate bound region is proportional to the constant term $C_{\text{total}}$, which itself depends on the bounds of the system disturbances and the estimation error of the learning mechanism. This concludes the proof.
\end{proof}

\subsection{Probabilistic Constraint Satisfaction}

The primary objective of the RAAR-MPC framework is to ensure that the system satisfies its operational constraints probabilistically over long horizons. This section provides a formal proof of this property, which synthesizes the probabilistic coverage of the LPES from Lemma \ref{lemma:lpes_coverage} with the closed-loop stability guarantees from Theorem \ref{thm:stochastic_stability}. The analysis is presented in two parts: first, a guarantee over the finite prediction horizon conditional on the learning outcome, and second, a proof of long-term convergence of the empirical risk to the desired level $\delta$.

\begin{theorem}[Probabilistic Constraint Satisfaction]
\label{thm:prob_satisfaction}
Consider the system (1) under the RAAR-MPC control law. Subject to Assumptions 1-5, the following properties hold:
\begin{enumerate}
    \item[\textit{i)}] (Finite-Horizon Guarantee) At any time $t$, if the optimization problem (8) admits a solution with zero slack, $\varepsilon_t=0$, then the state and input constraints are satisfied over the prediction horizon $\{0, \dots, N-1\}$ with a probability of at least $1 - \epsilon_{\text{LPES}}$. Formally,\(\mathbb{P} ( \forall k \in \{0, \dots, N-1\}: x_{t+k} \in\mathcal{X} \text{ and } u_{t+k} \in \mathcal{U} ) \ge 1 -\epsilon_{\text{LPES}}\).
   
    \item[\textit{ii)}] (Long-Term Risk Convergence) The long-term empirical frequency of physical constraint violations converges in probability to the user-specified risk level $\delta$. That is, for any $\nu > 0$:
    \begin{equation}
        \lim_{T \to \infty} \mathbb{P} \left( \left| \frac{1}{T}\sum_{t=0}^{T-1} \mathbb{I}(h(x_t)>0) - \delta \right| > \nu \right) = 0.
    \end{equation}
\end{enumerate}
\end{theorem}

\begin{proof}
The assertion of part (i) is that, given a feasible solution to the optimization problem (8) with $\varepsilon_t=0$, the resulting control action ensures that the system's state and input constraints are satisfied over the prediction horizon $k \in \{0, \dots, N-1\}$ with a joint probability of at least $1-\epsilon_{\text{LPES}}$.

Let the true, but unknown, sequence of uncertainties over the horizon be denoted by $\zeta \in \mathcal{U}^N$. The closed-loop system dynamics, when initiated from state $x_t$ with the optimal plan from (8), will produce a prediction error trajectory $\{e_{k|t}(\zeta)\}_{k=0}^{N-1}$. The satisfaction of the state constraint $x_{t+k} \in \mathcal{X}$ and input constraint $u_{t+k} \in \mathcal{U}$ is guaranteed if the error $e_{k|t}$ is contained within the total uncertainty set $\mathcal{U}_{\text{total},k}(t) = \mathcal{S}_k(t) \oplus \mathcal{B}(\beta_t)$. Since the adaptive safety margin corresponds to a non-negative buffer ($\beta_t \ge 0$), a sufficient condition for constraint satisfaction for all $k$ is the joint event that the error trajectory is fully contained within the sequence of Learned Prediction-Error Sets (LPES). Let this event be denoted by $\mathcal{C}_{\text{joint}}$:
\begin{equation}
    \mathcal{C}_{\text{joint}}(\zeta) \triangleq \bigcap_{k=0}^{N-1} \left\{ e_{k|t}(\zeta) \in \mathcal{S}_k(t) \right\}.
\end{equation}
Our goal is to establish a lower bound on the probability of this event, $\mathbb{P}_{\zeta \sim P(\mathcal{U}^N)}[\mathcal{C}_{\text{joint}}(\zeta)]$.

We proceed by analyzing the probability of the complement event, $\neg\mathcal{C}_{\text{joint}}$, which corresponds to the occurrence of at least one coverage failure. This can be expressed as the union of single-step failure events:
\begin{equation}
    \neg\mathcal{C}_{\text{joint}}(\zeta) = \bigcup_{k=0}^{N-1} \left\{ e_{k|t}(\zeta) \notin \mathcal{S}_k(t) \right\}.
\end{equation}
By the union bound (Boole's inequality), the probability of this event is bounded by the sum of the probabilities of the constituent events:
\begin{equation}
    \mathbb{P}(\neg\mathcal{C}_{\text{joint}}) \le \sum_{k=0}^{N-1} \mathbb{P}\left(e_{k|t}(\zeta) \notin \mathcal{S}_k(t)\right).
    \label{eq:formal_union_bound}
\end{equation}
The term $\mathbb{P}(e_{k|t}(\zeta) \notin \mathcal{S}_k(t))$ denotes the marginal probability that a coverage failure occurs at prediction step $k$. As established in the proof of Lemma 1, the GP-based active learning framework is designed such that the probability of failing to identify a critical uncertainty scenario that would lead to a coverage violation is bounded. Specifically, for any desired single-step risk level $\epsilon'$, it is possible to configure the learning algorithm (i.e., select $N_{\text{cand}}$, $t_{\min}$, and the UCB parameters) such that for any step $k$:
\begin{equation}
    \mathbb{P}\left(e_{k|t}(\zeta) \notin \mathcal{S}_k(t)\right) \le \epsilon'.
\end{equation}
To achieve a joint probabilistic guarantee of $1-\epsilon_{\text{LPES}}$ over the horizon of length $N$, we can allocate the total risk budget $\epsilon_{\text{LPES}}$ uniformly across the $N$ steps. We thus set the required single-step reliability by choosing $\epsilon' = \epsilon_{\text{LPES}}/N$. Substituting this into the inequality \eqref{eq:formal_union_bound} yields:
\begin{equation}
    \mathbb{P}(\neg\mathcal{C}_{\text{joint}}) \le \sum_{k=0}^{N-1} \frac{\epsilon_{\text{LPES}}}{N} = \epsilon_{\text{LPES}}.
\end{equation}
This implies that the probability of the joint success event $\mathcal{C}_{\text{joint}}$ is bounded as follows:
\begin{equation}
    \mathbb{P}(\mathcal{C}_{\text{joint}}) = 1 - \mathbb{P}(\neg\mathcal{C}_{\text{joint}}) \ge 1 - \epsilon_{\text{LPES}}.
\end{equation}
Since the event $\mathcal{C}_{\text{joint}}$ is a sufficient condition for constraint satisfaction over the entire horizon, we have thus formally shown that $\mathbb{P}(\forall k \in \{0, \dots, N-1\}: x_{t+k} \in \mathcal{X} \land u_{t+k} \in \mathcal{U}) \ge 1 - \epsilon_{\text{LPES}}$. This concludes the proof of Theorem 4(i). 

To prove part \textit{ii)}, we analyze the stationary behavior of the stochastic approximation algorithm governing $\beta_t$. The update law (25) is designed to find a root of the expected value of its driving term, $e_{\text{sa}}(t) = \mathbb{I}(L_t) - \delta_l(t)$. The stationary points of the adaptation dynamics are characterized by the condition $\mathbb{E}[e_{\text{sa}}(t)|\mathcal{F}_t] = 0$. This gives:
\begin{equation}
    P(L_t|\beta_t) - (\delta + \hat{P}_t(B_t)) = 0.
\end{equation}
The stability and resulting ergodicity of the joint process $(x_t, \beta_t)$, established in Theorem \ref{thm:stochastic_stability}, implies that the system converges to a stationary distribution. Under this stationary distribution, the law of large numbers ensures that the estimator $\hat{P}_t(B_t)$ converges to the true conditional probability $P(B_t|\beta_t)$ for a sufficiently large window size $W$. The equilibrium condition thus becomes:
\begin{equation}
    P(L_t|\beta_t) - (\delta + P(B_t|\beta_t)) = 0.
\end{equation}
By decomposing the probability of the learning event, $P(L_t|\beta_t) = P(h(x)>0|\beta_t) + P(B_t|\beta_t)$, the equation simplifies to:
\begin{equation}
    P(h(x)>0|\beta_t) = \delta.
\end{equation}
This shows that the equilibrium point of the self-correcting mechanism corresponds to the state where the true physical constraint violation probability equals the target risk level $\delta$.

The convergence of $\beta_t$ to this equilibrium point $\beta^*$ is guaranteed by the properties of the update law. It constitutes a Robbins-Monro stochastic approximation scheme. The stability of the underlying physical system ensures that the process noise is bounded. This, combined with the strict monotonicity condition in Assumption \ref{assum:adaptive_mech}(ii), which provides a negative feedback structure, satisfies the conditions for convergence of such algorithms, as established in classical results like those of Kushner and Clark \cite{kushner2012stochastic}. Therefore, $\beta_t$ converges to a neighborhood of $\beta^*$.

Finally, the Birkhoff ergodic theorem applies to the stationary and ergodic process $\{\mathbb{I}(h(x_t)>0)\}$. This theorem states that the time average of this indicator function converges almost surely to its expectation under the stationary measure. Since the dynamics of $\beta_t$ ensure this stationary expectation is $\delta$, we have:
\begin{equation}
    \frac{1}{T}\sum_{t=0}^{T-1} \mathbb{I}(h(x_t)>0) \xrightarrow{\text{a.s.}} \mathbb{E}_{\pi}[\mathbb{I}(h(x)>0)] = \delta,
\end{equation}
where $\mathbb{E}_{\pi}$ denotes the expectation with respect to the stationary distribution and $\xrightarrow{\text{a.s.}}$ denotes almost sure convergence. Convergence in probability is a direct consequence of almost sure convergence.
\end{proof}

\section{Numerical Example}
\label{sec:numerical_example}

In this section, we evaluate the performance of the proposed Risk-Aware Adaptive Robust MPC (RAAR-MPC) framework. We aim to demonstrate its ability to precisely manage constraint satisfaction under non-stationary uncertainties while maintaining low operational costs. The framework is benchmarked against three state-of-the-art methods on a challenging DC-DC converter control problem, characterized by both parametric model uncertainty and time-varying disturbances.

\subsection{Simulation Setup}
We consider the linearized discrete-time model of a DC-DC converter, a widely used benchmark in stochastic and robust control literature \cite{lorenzen2017constraint,capone2025online}. The system configuration is detailed as follows:
\begin{itemize}
    \item System Dynamics: The nominal model is described by
    \begin{equation}
        A = \begin{bmatrix} 1 & 0.0075 \\ -0.143 & 0.996 \end{bmatrix}, \quad B = \begin{bmatrix} 4.798 \\ 0.115 \end{bmatrix}.
    \end{equation}
    
    \item Cost Function: The objective is to minimize a standard quadratic cost $l(x, u) = x^T Q x + u^T R u$, with weighting matrices $Q = \text{diag}([1, 10])$ and $R = 1$. The terminal cost matrix $P$ is obtained by solving the discrete-time algebraic Riccati equation.
    
    \item Constraints: The system is subject to a hard input constraint $|u_t| \le 0.2$ and a chance constraint on the first state variable, given by:
    \begin{equation}
        \text{Pr}\left( [1, 0]x_t \le 0 \right) \ge 1 - \delta,
    \end{equation}
    where $\delta$ is the user-defined risk level. We test for a range of $\delta$ values to assess performance across different risk tolerances.
    
    \item Uncertainty Formulation: Our simulation setup introduces two significant challenges to emulate realistic operational conditions:
    \begin{enumerate}
        \item \textit{Parametric Uncertainty:} The true system matrices $A_{\text{true}}$ and $B_{\text{true}}$ deviate from the nominal model used by the MPC, with a relative uncertainty of up to 5\% for each entry.
        \item \textit{Non-Stationary Disturbances:} The additive disturbance $w_t$ is drawn from a uniform distribution whose bounds change over time, creating a non-stationary environment. The simulation of $110,000$ steps is divided into five epochs with different disturbance scaling factors, ranging from severe ($2.5\times$) to mild ($0.5\times$), based on a baseline of $[-0.14, 0.14]$. This setup rigorously tests the controller's ability to adapt to changing operational conditions.
    \end{enumerate}
    
    \item RAAR-MPC Configuration: The proposed controller uses a prediction horizon of $N=10$. The risk assessment engine (LPES) is updated every $T_{\text{update}}=50$ steps, using $K_{\text{crit}}=10$ critical scenarios. The self-correcting risk regulation loop for the adaptive safety margin $\beta_t$ is configured with a primary learning rate $\alpha_0=0.05$ and a much smaller mean-reversion rate $\gamma_0=0.0001$, ensuring that risk tracking dominates while preventing drift. The history window for buffer probability estimation is set to $W=100$.
\end{itemize}

\subsection{Illustration of the RAAR-MPC Adaptation Mechanism}

To visualize the internal workings and demonstrate the effectiveness of our proposed RAAR-MPC framework, we conduct a detailed simulation run with a fixed target risk of $\delta = 0.1$. The system operates under the challenging non-stationary disturbance profile described in Section V-A. 

First, we present the overall closed-loop system performance in Figure~\ref{fig:system_response}. The top plot shows the evolution of the constrained state $x_{1,t}$. The trajectory's variance clearly correlates with the disturbance intensity across the five different epochs, demonstrating the controller's ability to react to changing conditions. The zoomed-in view highlights that constraint violations (red dots) are not eliminated but are carefully managed to occur at a frequency consistent with the target risk level. The bottom plot confirms that the control input $u_t$ remains strictly within its hard bounds throughout the entire simulation. This figure establishes that the framework successfully controls the system while adhering to its operational constraints in a probabilistic sense.

\begin{figure}[thpb]
    \centering
    \includegraphics[width=1.0\columnwidth]{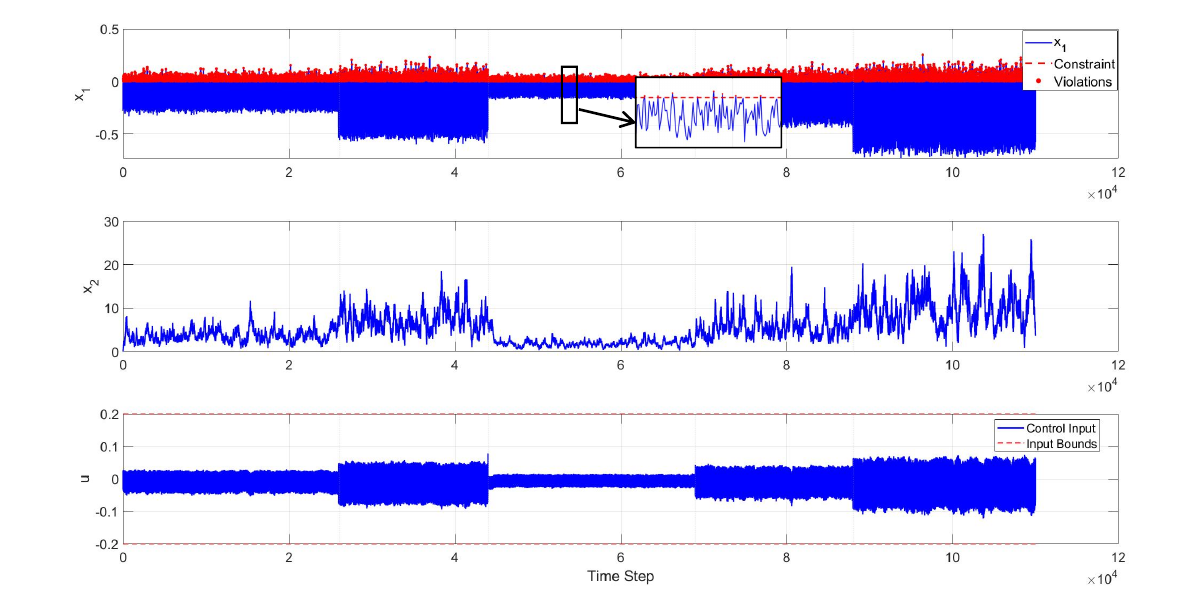} 
    \caption{Closed-loop system response under RAAR-MPC control for a target risk $\delta=0.1$. The figure shows the evolution of the states ($x_1, x_2$) and the control input ($u$) over 110,000 steps with non-stationary disturbances. The controller successfully manages the state constraints while respecting input bounds.}
    \label{fig:system_response}
\end{figure}

The key to this robust and efficient performance lies in the framework's dual-adaptive mechanism, which operates on two different timescales. The dynamic evolution of its two core components is illustrated in Figure~\ref{fig:adaptation_components}.
\begin{itemize}
    \item Proactive Adaptation (LPES):The top plot shows the size of the Learned Prediction-Error Set (LPES). This proactive component, updated by the medium-frequency learning loop, dynamically expands during high-disturbance periods and contracts when conditions are mild. This demonstrates the engine's ability to construct a tight, data-driven characterization of the predictable uncertainty.
    \item Reactive Adaptation (Safety Margin): The middle plot shows the evolution of the adaptive safety margin $\beta_t$. This reactive component, adjusted by the low-frequency, experience-driven loop, increases to enhance robustness when disturbances are high and decreases to reduce conservatism when they are low.
\end{itemize}
The bottom plot shows the net constraint tightening, which is the combined effect of both components. This synergy allows the system to mount a robust defense against uncertainty: the LPES handles the learned, structural part, while the safety margin provides a fast-reacting buffer against unmodeled or transient effects.

\begin{figure}[thpb]
    \centering
    \includegraphics[width=1.0\columnwidth]{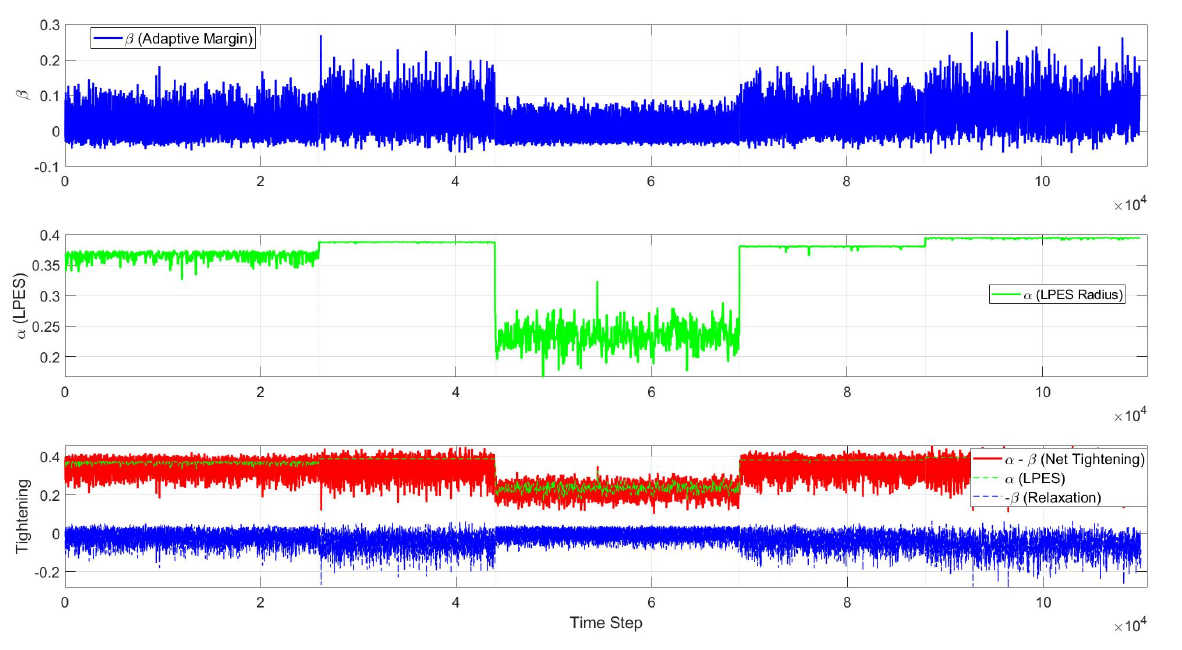} 
    \caption{Evolution of the dual-adaptive components. Top: The size of the proactive Learned Prediction-Error Set (LPES). Middle: The reactive Adaptive Safety Margin $\beta_t$. Bottom: The resulting net constraint tightening and its constituent parts.}
    \label{fig:adaptation_components}
\end{figure}

A critical innovation that enables the precise and robust performance of the reactive loop is the dynamic learning trigger, which addresses the "signal sparsity" problem inherent in learning from rare events. This mechanism is visualized in Figure~\ref{fig:safety_buffer}. Instead of waiting for rare physical violations ($h(x_t)>0$), our framework defines a "Safety Buffer" (the shaded green area) between the physical constraint and a dynamic learning boundary at $-m_s(t)$. A "learning event" is triggered whenever the state enters this buffer. Since these events are far more frequent than physical violations, they provide a rich and persistent feedback signal to the stochastic approximation algorithm governing $\beta_t$. This allows the controller to precisely regulate the empirical risk towards the desired level $\delta$ without delay.

\begin{figure}[thpb]
    \centering
    \includegraphics[width=1.0\columnwidth]{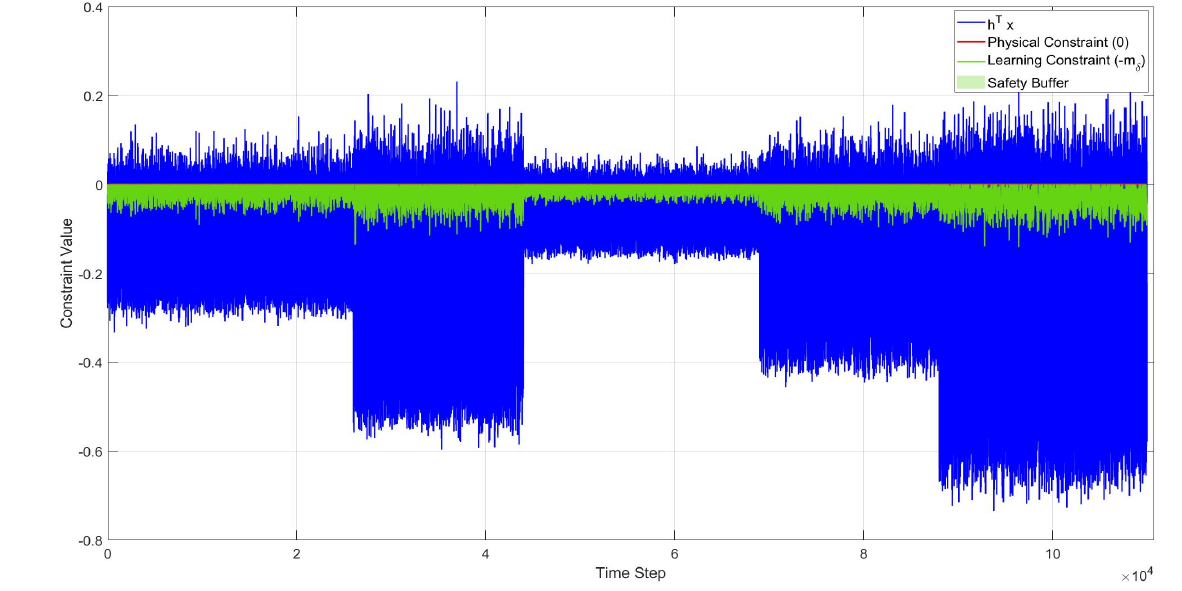}
    \caption{The self-correcting risk regulation mechanism. The "Safety Buffer" (shaded green) between the physical and learning constraints provides a frequent trigger for the adaptation of the safety margin $\beta_t$.}
    \label{fig:safety_buffer}
\end{figure}

In summary, these illustrations demonstrate that the successful overall system control shown in Figure~\ref{fig:system_response} is a direct result of a cohesive internal mechanism. The dual-adaptive components in Figure~\ref{fig:adaptation_components} work in synergy, and the effectiveness of the reactive component is greatly enhanced by the novel learning trigger concept shown in Figure~\ref{fig:safety_buffer}.

\subsection{Comparison With Existing Approaches}

We now benchmark our RAAR-MPC against other prominent control strategies. We select two state-of-the-art methods for comparison: the sampling-based approach by Lorenzen et al. \cite{lorenzen2017constraint} and the online estimation approch by Capone et al. \cite{capone2025online}. 

First, we evaluate the overall performance across seven different target risk levels $1-\delta \in \{0.6, 0.7, 0.8, 0.9, 0.95, 0.99\}$. The simulations are run for 110,000 time steps under the non-stationary disturbance setting. The aggregated results are summarized in Figure~\ref{fig:comparison}.

\begin{figure*}[t]
    \centering
    \begin{subfigure}[b]{0.48\textwidth}
        \centering
        \includegraphics[width=\textwidth]{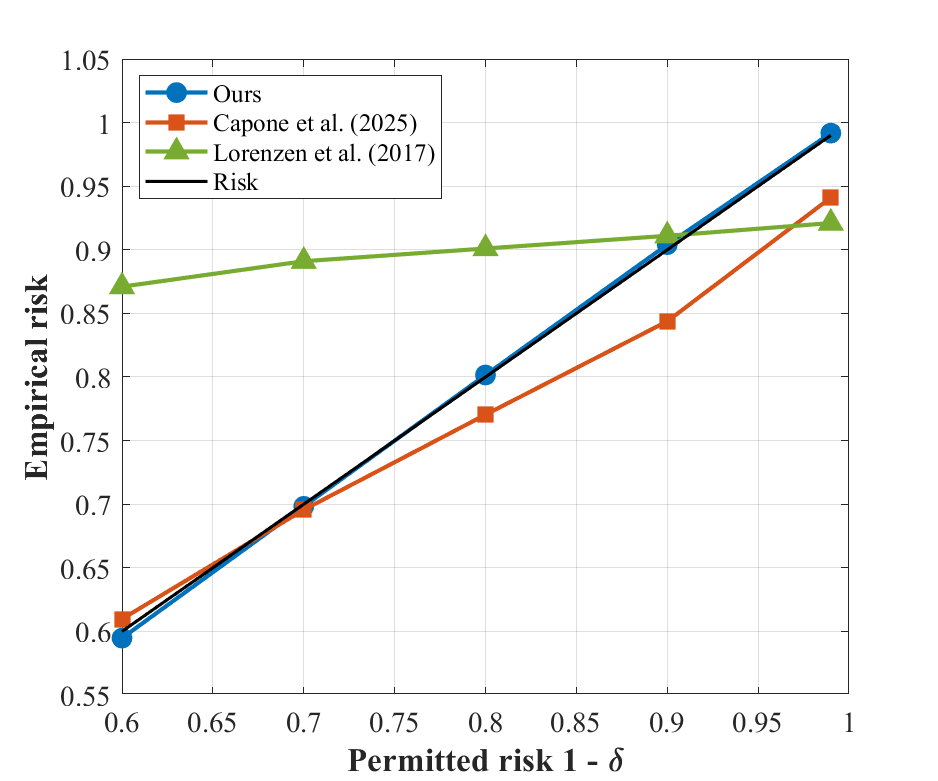}
        \caption{Empirical risk vs. permitted risk ($1-\delta$).}
        \label{fig:comparison_risk}
    \end{subfigure}
    \hfill
    \begin{subfigure}[b]{0.48\textwidth}
        \centering
        \includegraphics[width=\textwidth]{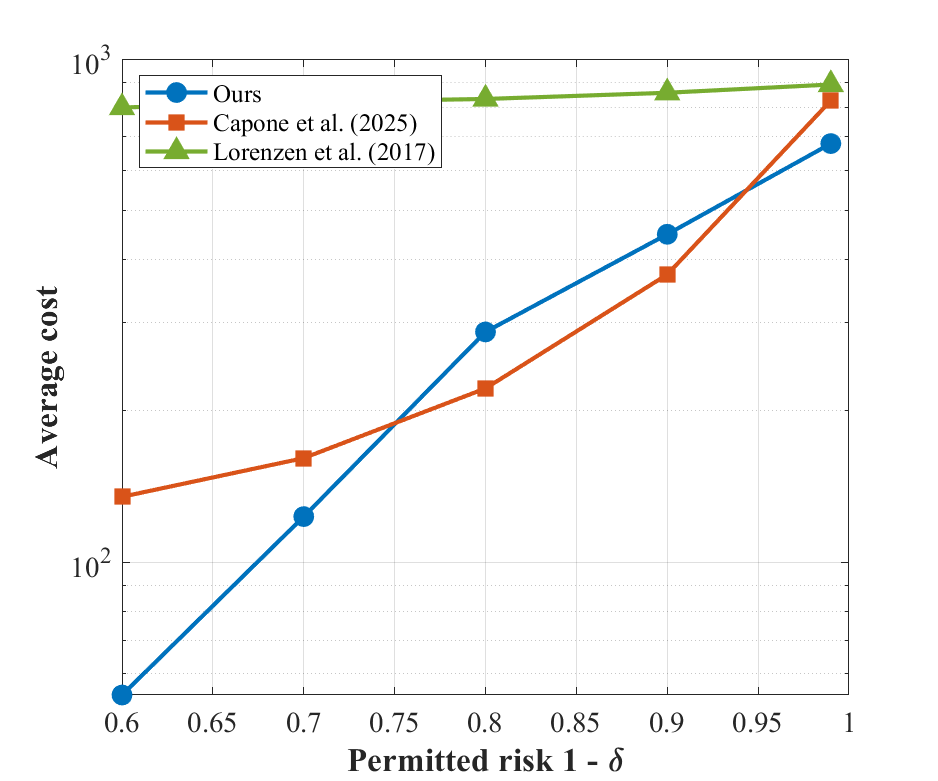}
        \caption{Average cost vs. permitted risk ($1-\delta$).}
        \label{fig:comparison_cost}
    \end{subfigure}
    \caption{Comparative performance analysis across different risk preferences. (a) This plot shows that our RAAR-MPC (blue circles) precisely tracks the target risk level (black dashed line). (b) This plot reveals that our precise risk tracking leads to significantly lower average costs.}
    \label{fig:comparison}
\end{figure*}

Figure~\ref{fig:comparison_risk} compares the empirical rate of constraint satisfaction against the permitted rate. Our RAAR-MPC (blue circles) demonstrates exceptional precision, with its data points lying almost perfectly on the target line, a direct result of the self-correcting adaptive margin $\beta_t$. In contrast, the other methods show conservative behavior. Figure~\ref{fig:comparison_cost} reveals the significant performance advantage of this precision. By avoiding unnecessary conservatism, our RAAR-MPC achieves a substantially lower average operational cost, especially at higher safety requirements (higher $1-\delta$).

To further investigate the robustness and adaptability of our framework—a crucial aspect for non-stationary environments—we conduct a more granular analysis. Figure~\ref{fig:robustness_comparison} shows the performance of each method under the five different disturbance epochs (represented by their scale) for a fixed target risk of $\delta=0.01$ (requiring 99\% satisfaction).

\begin{figure}[thpb]
    \centering
    \includegraphics[width=0.95\columnwidth]{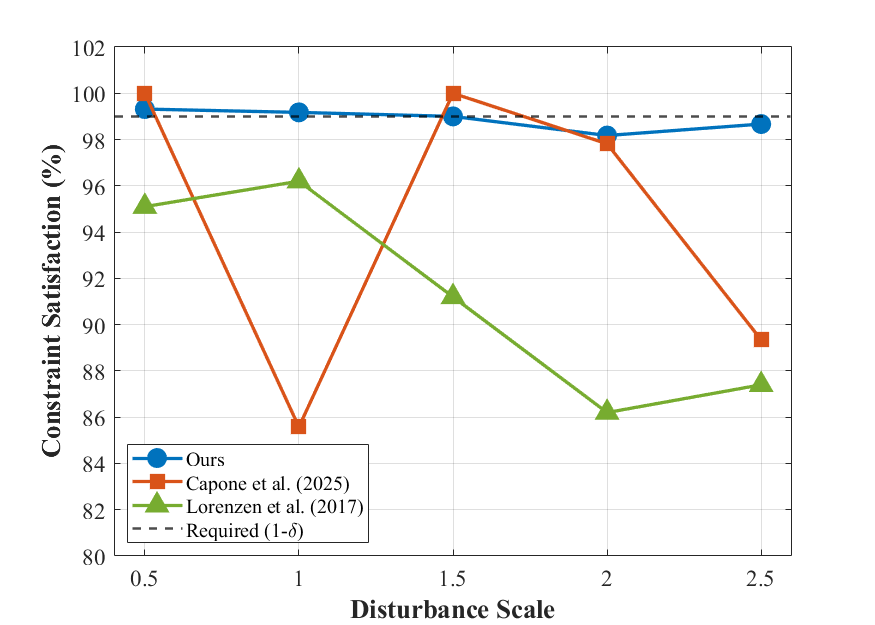}
    \caption{Robustness and adaptability comparison under varying disturbance conditions for a fixed target risk $\delta=0.01$. Our method (blue circles) consistently achieves the required 99\% satisfaction rate (dashed line), while competing methods exhibit significant deviations, indicating a lack of adaptability.}
    \label{fig:robustness_comparison}
\end{figure}

The results in Figure~\ref{fig:robustness_comparison} are striking. Our RAAR-MPC demonstrates remarkable robustness, maintaining the empirical satisfaction rate almost exactly at the required 99\% level across all disturbance scales. This is a direct consequence of the dual-adaptive mechanism, which effectively learns and compensates for the changing uncertainty characteristics in real time. In stark contrast, the performance of the competing methods is erratic. The method of Capone et al. \cite{capone2025online} is overly conservative at low disturbance but becomes significantly unsafe (satisfaction drops to ~85\%) when the disturbance profile changes. The approach of Lorenzen et al. \cite{lorenzen2017constraint} is consistently too aggressive (i.e., it violates the constraint more often than permitted) and its performance further deteriorates under more severe disturbances. This highlights their inability to effectively adapt their safety margins to the changing operational reality.

A crucial advantage of our RAAR-MPC is therefore evident: the dual-timescale adaptation provides a uniquely robust response. The low-frequency, reactive update of the safety margin $\beta_t$ responds almost instantaneously to any observed change in constraint satisfaction statistics, providing an immediate first line of defense. The medium-frequency LPES update then follows to characterize the new disturbance regime more accurately. This superior adaptability, directly evidenced by the stable performance across varying disturbance scales shown in Figure~\ref{fig:robustness_comparison}, ensures robust and efficient performance even during rapid transitions in the operating environment.

In summary, the comparative analysis clearly indicates that the proposed RAAR-MPC framework achieves a superior trade-off between safety and performance. It not only meets the specified chance constraints with high precision across various risk preferences, but also demonstrates exceptional robustness and adaptability under challenging non-stationary conditions, all while operating at a lower control cost.

\section{Conclusion}
\label{sec:conclusion}
In this paper, we have presented the Risk-Aware Adaptive Robust MPC (RAAR-MPC), a novel framework for controlling constrained linear systems subject to significant, non-stationary uncertainties. The key innovation lies in a dual-layer, multi-timescale architecture that decouples intelligent online risk assessment from the real-time control task, while linking them through a synergistic adaptive mechanism.

The framework employs a medium-frequency risk assessment engine, which leverages Gaussian processes and active learning to proactively identify critical uncertainty scenarios. This process constructs a tight, data-driven Learned Prediction-Error Set (LPES), which effectively reduces the conservatism inherent in traditional worst-case methods. Concurrently, a low-frequency, experience-driven risk regulation loop adjusts an adaptive safety margin based on closed-loop performance statistics. This ensures that the system precisely meets the user-defined chance constraint probability, robustly compensating for unmodeled dynamics and non-stationarities.

We have formally established the key theoretical properties of the RAAR-MPC framework, guaranteeing recursive feasibility by construction, closed-loop stability of the augmented system state (including the adaptive margin), and convergence of the empirical violation rate to the target risk level with high probability. The practical efficacy and superiority of our approach were demonstrated through extensive numerical simulations on a benchmark DC-DC converter under challenging non-stationary disturbance conditions. The results show that the RAAR-MPC not only achieved precise risk tracking across various risk levels but also did so at a significantly lower average operational cost compared to other state-of-the-art robust and stochastic control strategies.

Future work will proceed along two primary directions. First, extending the framework to handle uncertain nonlinear systems is a key objective. This will require the development of new methods for characterizing nonlinear prediction errors and establishing stability for the coupled physical and learning dynamics. Second, we aim to provide a more detailed theoretical analysis, including formal quantitative bounds on the convergence rate of the adaptive margin and the size of its ultimate invariant set, further strengthening the performance guarantees of the proposed approach.

\bibliographystyle{IEEEtran}
\bibliography{ref}

\end{document}